\documentclass[twocolumn]{aastex631}
\RequirePackage{amsthm,amsmath,amsfonts,amssymb}

\usepackage{mathtools,breqn,comment,graphicx}
\usepackage{xcolor,epstopdf}
\usepackage[caption=false]{subfig}

\theoremstyle{plain}
\newtheorem{theorem}{Theorem} 
\newtheorem{lemma}[theorem]{Lemma}
\newtheorem{corollary}[theorem]{Corollary}

\newtheorem{conjecture}{Conjecture}

\theoremstyle{definition}

\theoremstyle{remark}

\usepackage{amsmath,amsthm,mathrsfs}
\usepackage{graphicx,psfrag,epsf}
\usepackage{enumerate}
\usepackage{natbib}
\usepackage{url} 
\usepackage{amssymb,amsmath}
\usepackage{color,xcolor,nicefrac}

\def\DeltaC{$\Delta {C}$}
\def\DeltaCSysEq{\Delta {C}_{\mathrm{sys}}}
\def\DeltaCSys{$\DeltaCSysEq$}

\def\ovii{O~VII}
\def\1es{1ES~1553+113}
\def\dof{degrees--of--freedom}

\def\apj{Astrophys. J.}
\def\cmin{$C_{\mathrm{min}}$}
\def\cmineq{C_{\mathrm{min}}}

\def\cminsys{$C_{\mathrm{min,sys}}$}

\def\DeltaC{$\Delta C$}
\newcommand\Var{\text{Var}}

\newcommand\E{\text{E}}

\def\Laplace{\mathcal{L}}
\def\erf{\text{erf}}

\def\sigmaIntRelEq{\sigma_{\mathrm{int},i}/\hat{\mu}_i}

\def\paperOne{Paper~I}

\def\sigmaIntiEq{\sigma_{\mathrm{int},i}}

\def\sigmaIntRelEq{\sigma_{\mathrm{int},i}/\hat{\mu}_i}

\def\gof{goodness--of--fit}

\shorttitle{Poisson regression with systematic errors II}
\shortauthors{Bonamente et al.}
\begin{document}

  \title{Maximum--likelihood regression with systematic errors for astronomy and the physical sciences:\\
  II. Hypothesis testing of nested model components for Poisson data}

\author[0000-0002-8597-9742]{Massimiliano Bonamente}\affiliation{Department of Physics and Astronomy, University of Alabama in Huntsville, Huntsville, AL 35899}
\author[0000-0003-1212-4089]{Dale  Zimmerman}\affiliation{Department of Statistics and Actuarial Science,
University of Iowa, Iowa City, IA 52242}
\author[0000-0002-9516-8134]{Yang Chen}\affiliation{Department of Statistics, University of Michigan, Ann Arbor, MI 48109}


\begin{abstract}
A novel model of systematic errors for the regression of Poisson data is applied to hypothesis testing
of nested model components with the introduction of a generalization of the 
\DeltaC\ statistic that applies in the presence of systematic errors. 
This paper shows that the null--hypothesis parent distribution of this \DeltaCSys\ statistic can be obtained 
either through a simple numerical procedure, or in a closed form
by making certain simplifying assumptions. It is found that the effects of systematic errors
on the test 
statistic can be significant, 
and therefore the inclusion of sources of systematic errors is crucial
for the assessment of the significance of nested model component in practical applications. The methods proposed in this paper provide a simple and accurate means of 
including systematic errors for hypothesis testing of nested model components in a
variety of applications.
\end{abstract}



\section{Introduction}

The regression of integer--count Poisson data to a parametric model occupies a central
role in the analysis of data for astronomy and the physical sciences \citep[see, e.g.,][]{james2006}. 
The \gof\ statistic
for hypothesis testing and parameter estimation is usually referred to as the Poisson \emph{deviance}  in the statistical
literature \citep[e.g.][]{cameron2013, bishop1975, goodman1969},
and it is generally known as the \emph{Cash} statistic \cmin\ for astronomical applications 
\citep{cash1976, cash1979, baker1984}. The asymptotic distribution of the fit statistic under the null hypothesis
that the parametric model is correct is available in the large--count regime, where it is 
distributed as a chi--squared variable; 
and in the extensive data regime where, regardless of number of counts, it is asymptotically distributed 
like a normal variable \citep{li2024}. Therefore, for most data analysis cases, the regression
of Poisson data and the associated parameter estimation and hypothesis testing can be performed without complications.

A new method for the inclusion of sources of systematic uncertainty in the maximum--likelihood 
regression of Poisson
data to a parametric model was introduced in a companion paper \citep{bonamente2024}, hereafter referred to as \paperOne. The method consists of introducing an \emph{intrinsic model variance} that characterizes
systematic uncertainties in the best--fit model. The main advantages of this method are the ability 
to use the familiar \cmin\ statistic to obtain the best--fit model parameters \citep{cash1976, cash1979}, and then
generalize the \gof\ statistic to a new statistic \cminsys\ that has simple analytic properties.
As a result, hypothesis testing in the presence of systematic error becomes a simple task that can be accomplished
for virtually any integer--count Poisson data sets, same as for the original \cmin\ statistic.

Another common data analysis task is the assessment of the significance of a nested model component.
This is particularly common for astronomy and the physical sciences, where certain parametric components
represent a physically--motivated modification of a baseline model, such as the presence
of an emission or absorption line superimposed to an underlying continuum \citep[e.g.][]{nicastro2018, spence2023},
or the presence of an exponential cutoff at high energy that modifies
a continuum emission mechanism \citep[e.g.][]{tang2015}. 
Following the statistical method of systematic errors presented in \paperOne, this paper focuses
on the distribution of the associated \DeltaC\ statistic for hypothesis testing
of nested model components in the presence of systematic errors, which is hereafter referred to as \DeltaCSys. 

{The statistics literature provides a much broader context for the
regression with count data, beyond the simple equidispersed Poisson distribution
considered in this paper. In particular, alternative methods of regression with overdispersed distributions such as the negative binomial \citep[e.g.][]{hilbe2011,hilbe2014}, or the Poisson--inverse--Gaussian \citep[e.g.][]{tweedie1957, sichel1971} distributions were reviewed in Sec.~3.2 of \paperOne. The choice to focus on the Poisson distribution was made primarily because of its relative simplicity and the availability of a \gof\ statistic that is tractable, as discussed in \paperOne\ and also reviewed below in Sec.~\ref{sec:paperOneSummary} of this paper.}

This paper
is structured as follows: Sec.~\ref{sec:DeltaCIntro} describes the  \DeltaCSys\ statistic, including a review of certain results from \paperOne. Sec.~\ref{sec:DeltaY}
discusses the distribution of the $\Delta Y$ component of the $\Delta C$ statistic that is introduced 
by the presence of systematic errors. Sec.~\ref{sec:DeltaC} then provides the distribution
for \DeltaCSys, including
 Monte Carlo simulations and 
analytical approximations.
Sec.~\ref{sec:HypothesisTesting} provides a discussion of hypothesis testing for a nested model
component with the $\Delta C$ statistic and a case study with real--life astronomical data.
Conclusions are provided in Sec.~\ref{sec:conclusions}.

\section{The $\Delta C$ statistic for nested model components}
\label{sec:DeltaCIntro}

A common problem in statistical data analysis is assessing the significance of a nested model component in
the maximum--likelihood regression to a parametric model. In this paper we focus on Poisson regression, with a
\gof\ statistic \cmin, and  \DeltaC\ as the statistic of choice for nested model components. This section
provides a definition of the relevant statistics and a brief overview of the results from \paperOne\ 
for the inclusion of systematic errors.

\subsection{The Poisson deviance \cmin\ and the $\Delta C$ statistic}

The data model is of the type $(x_i, y_i)$, for $ i=1,\dots,N$ independent
Poisson--distributed measurements $y_i$ at different values of an independent variable $x_i$,
where $\theta=(\theta_1,\dots,\theta_m)$ are the $m$ adjustable parameters of the model,
and $\hat{\mu}=(\hat{\mu}_1,\dots,\hat{\mu}_N)$ the means of the Poisson distributions evaluated at the best--fit parameter values, same as in \paperOne.

In the regression of integer--count Poisson data to a parametric model, it is customary to use the deviance
\begin{equation}
  D_P=2 (\mathcal{L}(y)-\mathcal{L}(\hat{\mu})) = \sum_{i=1}^N \left( y_i \ln \left(\dfrac{y_i}{\hat{\mu}_i}\right) - (y_i-\hat{\mu}_i) \right)
  \label{eq:DP}
\end{equation}
as the goodness--of--fit statistic, which is twice
the difference between the maximum achievable log--likelihood $\mathcal{L}(y)$ and that of the fitted model. This statistic is known as the Cash statistic
and usually indicated as $\cmineq \coloneqq D_P$. This \gof\ statistic
 was spearheaded by \cite{cash1976} and \cite{cash1979} and others \citep[e.g.][]{baker1984} in response
to a wealth of new high--energy astronomy data that were collected by photon--counting devices,
starting with the early satellite--based X--rays surveys of the 1970's \citep[e.g.][]{giacconi1971,giacconi1972}.

This paper extends the statistical model defined in \paperOne\ to the statistic
\begin{equation}
  \Delta D_P = D_P(\theta^T_{k},\hat{\theta}_{m-k}) - D_P(\hat{\theta}_m) \coloneqq \Delta C
  \label{eq:DeltaDP}
\end{equation}
where $ D_P(\hat{\theta}_m)$ is the usual $D_P$ statistic where all of the $m$ parameters are fit to the
data, as in \eqref{eq:DP}, while $D_P(\theta^T_{k},\hat{\theta}_{m-k})$ is evaluated with $k \leq m$ parameters
held fixed at the true--yet--unknown parent values. The $\Delta D_P$ statistic can be written
in a simplified notation as
\begin{equation}
    \Delta C = C_{\mathrm{min,r}}-C_{\mathrm{min}},
    \label{eq:DeltaC}
\end{equation} 
and it is  hereafter 
referred to as  $\Delta C$. 
Tests for the significance of a nested model component with $k$ additional parameters in the ML regression with Poisson data are performed via this statistic, 
where 'r' labels the reduced model with fewer adjustable parameters (i.e., $m-k$ parameters), and for simplicity the full model 
with $m$ parameters will be hereafter indicated without subscripts.
The $\Delta C$ statistic is commonly used in physics and astronomy
applications, where a majority of data are photon counts that are conveniently modelled
by a Poisson process of constant rate \citep{cash1979}.

\subsection{Nested components and likelihood--ratio statistics}
A nested model component with $k \geq 1$ parameters is defined as a portion of the full model with $m$ parameters, 
with the feature that the full model becomes the reduced model by a suitable choice (or null values) of the $k$
parameters, often setting them to zero or to another fixed value. For example, a broken--power law model 
is a full model with $m=4$ parameters (e.g., a normalization, two power--law indices and the location of the break) that
becomes a power--law model with $m-k=2$ parameters with a suitable choice of the $k=2$ parameters in the
nested component (the second power--law index and the break). 

In the absence of systematic errors, the Wilks theorem  \citep[see][]{wilks1938,wilks1943,rao1973}
guarantees that $\Delta C \sim \chi^2(k)$.
There are a number of mathematical conditions
for the applicability of the null--hypothesis distribution of $\Delta C$, which is a likelihood--ratio statistic \citep[for a recent review, see][or \paperOne]{li2024}. Of particular relevance to applications in the physical sciences and astronomy is
the topology of parameter space, in particular the requirement that
the null values of the nested model component \emph{cannot} be at the boundary
of the allowed parameter space. 

A typical example of a parameter whose null value is at the boundary of
the allowed parameter space  is an emission (or absorption) line which is only allowed, respectively,
 a positive or a negative normalization. In this case, the Wilks theorem does not apply \citep[e.g.][]{chernoff1954}, as illustrated in an astrophysical context by \cite{protassov2002}. On the other hand,
 a model for an absorption/emission line that allows both positive and negative normalization
 \citep[e.g., as illustrated in][]{spence2023} follows the required topological constraints, as it does
 a typical parameterization of a broken power--law model.


\subsection{Review of key results from \paperOne}
\label{sec:paperOneSummary}

In order to provide a self--contained description of the present results for the $\Delta C$ statistic,
a brief summary of the key results from \paperOne\ for the \cmin\ statistic is presented in this section.
Systematic errors are modelled via an intrinsic model variance $\sigmaIntiEq^2$ which is defined as the variance of a random variable $M_i$ that describes the distribution of the best--fit model $\hat{f}_i(x_i)$, namely with 
\begin{equation}
\begin{cases}
  \E(M_i)=\hat{\mu}_i\\
  \Var(M_i)=\sigmaIntiEq^2,
  \end{cases}
  \label{eq:sigmaInt}
\end{equation}
where $\hat{\mu}_i$ is the usual best--fit value according to the ML regression, and  $f_i=\sigmaIntRelEq  \ll 1$ models the relative value of the 
systematic error.~\footnote{For an illustration, see Fig.~1 of \paperOne.}
The purpose of this model variance is to model fluctuations or overdispersion in the data by treating the best--fit value as a random variable, rather than a number. 
The shape of the distribution for $M_i$ can be chosen at will, with the positive--valued gamma distribution being a reasonable choice, although the normal distribution can be used as well in most applications in the large--mean regime. { Notice how, for example, the compounding of a gamma with a Poisson distribution leads to a negative binomial distribution, as also explained in Sec.~3.2 of \paperOne, which provides the kind of overdispersion introduced by systematics that is envisioned by this method.}

{ We use a quasi--maximum likelihood method \citep[e.g.][]{cameron2013,gourieroux1984,gourieroux1984b} that 
retains the usual Poisson log--likelihood to estimate the parameters, and \eqref{eq:sigmaInt} are enforced \emph{post-facto} to determine the \gof\ statistic in the presence of systematics (see Sec.~4.1 of \paperOne).}
Accordingly,  we have
shown that 
the $\cmineq\coloneqq X+Y=Z$ statistic is the sum of two statistics, where $X\sim \chi^2(\nu)$  is 
the usual \cmin\ statistic in
the absence of systematic errors with $\nu=N-m$, and $Y$ is
the additional independent contribution due to systematic errors,
\begin{equation}
\begin{aligned}
    Y = 
    2 \sum_{i=1}^N (M_i-\hat{\mu}_i)-y_i \ln \left(\dfrac{M_i}{\hat{\mu}_i}\right),
    \end{aligned}
    \label{eq:YStat}
\end{equation}
which vanishes if $M_i$ is identically equal to $\hat{\mu}_i$, e.g., in the absence of systematic errors.
This statistic is
 is asymptotically distributed in the extensive regime  as
\begin{equation}
    Y\sim N(\hat{\mu}_C,\hat{\sigma}^2_{C}).
\end{equation}
The parameters $\hat{\mu}_C$ and $\hat{\sigma}^2_{C}$ are referred to as respectively the bias and overdisperion parameter, and they
 can be estimated from the data. 
The normal distribution for $Y$ applies regardless of choice for the $M_i$ random variable, provided the data are in the extensive regime, which is the assumption used throughout. The choice of distribution for $M_i$ only marginally affects the estimation of the
overdispersion parameter, and therefore it is not a crucial one.

In general, the distribution of $Z$ follows an \emph{overdispersed $\chi^2$ distribution} that is
the convolution of the pdf of the two (normal and chi--squared distributed)
constituting random variables, \citep[for properties, see][]{bonamente2024properties}. In the asymptotic limit of an extensive dataset in the
large--mean regime, both $X$ and $Y$ are normally distributed, resulting in
a normal distribution for $Z$, namely
\begin{equation}
Z \overset{a}{\sim} N(\nu+\hat{\mu}_C, 2 \nu+ \hat{\sigma}^2_C).    
\label{eq:ZAsymptotic}
\end{equation}
An underlying assumption is that $X$ and $Y$ are independent random variables, as was discussed in \paperOne.

\subsection{The $\Delta C$ statistic with systematic errors}

With the model of systematic errors described in \paperOne, the $\Delta C$ statistic in the presence of systematic errors
is modified to 
\begin{equation}
    \DeltaCSysEq \coloneqq (X_r-X) + (Y_r-Y)
    \label{eq:DeltaCsys}
\end{equation}
where $\Delta X \coloneqq(X_r-X) \sim \chi^2(k)$ represents the usual statistic according to \eqref{eq:DeltaC} (i.e., without the use of systematic errors) and the
additional term $\Delta Y \coloneqq (Y_r-Y)$ represents the additional contribution 
introduced by the systematic errors.  
The goal of this paper is to determine the distribution of this newly defined \DeltaCSys\ statistic.

\section{Distribution of the $\Delta Y$ statistic}
\label{sec:DeltaY}
This section studies the distribution of $\Delta Y$ and its relationship to $\Delta X$,
in order to determine the distribution of \DeltaCSys\ in the presence of systematic errors according to \eqref{eq:DeltaCsys}, which will then be studied in Sec.~\ref{sec:DeltaC}.

\subsection{General considerations}
\label{sec:DeltaYGeneral}

The effect of the random variable $M_i$ according to \eqref{eq:sigmaInt} on the $Y$ and $Y_r$ statistics is that of
a 
\emph{randomization} of the best--fit model according to \eqref{eq:sigmaInt},
which was also used in \paperOne. 
In the following, we use the notation $N(\mu, \sigma^2)$ to describe the distribution of the $M_i$ random variables,
with mean $\mu$ and variance $\sigma^2$ according to \eqref{eq:sigmaInt}, implying that a normal distribution is a suitable choice to model systematic errors. In \paperOne\ we described the effect of choice of distribution for $M_i$ (e.g., Gaussian, gamma or other), and concluded that the choice is not critical. Those conclusions also apply to the results to be presented in this paper,
and the effects of distributional choices for $M_i$ will be discussed where relevant throughout the paper. 

This means that the randomized
statistics are evaluated using, respectively,  
\begin{equation}
    \begin{cases}
        \mu_i \sim N(\hat{\mu}_i,\sigma^2_{\text{int},i})   \text{ with } \sigmaIntiEq= f \cdot \hat{\mu}_i \text{ (for } C_{\mathrm{min}})\\
        \mu_{r,i} \sim N(\hat{\mu}_{r,i},\sigma^2_{\text{int,r},i}) \text{ with }
  \sigma_{\text{int,r},i}= f \cdot \hat{\mu}_{r,i} \text{ (for } C_{\mathrm{min,r}}).
    \end{cases}
\end{equation}
 The constant $f$ represents the 
  relative value of the intrinsic model error, i.e.,
  \begin{equation} 
  f\coloneqq \dfrac{\sigma_{\text{int},i}}{\hat{\mu}_i} = \dfrac{\sigma_{\text{int,r},i}}{\hat{\mu}_{r,i}}
  \label{eq:f}
  \end{equation}
e.g., $f=0.1$ for a 10\% level of systematic errors. This value is assumed
constant for all data points, and it is expected that $f \ll 1$.

Accordingly, we set
\begin{equation}
\begin{cases}
    \mu_i = \hat{\mu}_i + x_i\, \text{ with } x_i \sim N(0, (f \cdot \hat{\mu}_i)^2)\\
    \mu_{r,i} = \hat{\mu}_{r,i} + x_{r,i} \, \text{ with } x_{r,i} \sim N(0, (f \cdot \hat{\mu}_{r,i})^2).
\end{cases}
    \label{eq:mui}
\end{equation}

Each pair of random variables $x_{r,i}$ and $x_i$ corresponding to the randomization
of the best--fit models in the $i$--th bin, however, are \emph{not} independent of one another by design. 
In fact,
they must be modelled as having \emph{perfect} correlation, in that they represent the occurrence that
a systematic error  causes a given datum to be shifted by a given amount, and therefore
both randomized models will follow the same random shift. This
perfect correlation between $x_{r,i}$ and $x_i$ results in the following expectations for this difference:
\begin{equation}
    \begin{cases}
        \E(x_{r,i}-x_i) = 0\\
        \Var(x_{r,i}-x_i)= f^2 \,(\hat{\mu}_{r,i}-\hat{\mu}_i)^2 ,
    \end{cases}
\end{equation}
with the consequence that
a data point where the full and reduced model are identical will feature a null contribution from the
randomization of the models.~\footnote{This is in contrast with the case where $x_{r,i}$ and $x_i$ are
independent, in which case $\Var(x_{r,i}-x_i)=f^2 \,(\hat{\mu}_{r,i}^2+\hat{\mu}_i^2)$ would apply. This
is not the case for this model of systematic errors.} 
Also,  it is immediate to show that
\begin{equation}
    \Delta Y \simeq 2\, \sum_{i=1}^N (x_{r,i}-x_i),
    \label{eq:DeltaY}
\end{equation}
the approximation holding when $f \ll 1$ (see Appendix~\ref{app:randomization} for details).  This property applies
to any parameterization of the models.

The correlation between each $x_{r,i}$ and $x_i$ pair has therefore an effect on the distribution
of $\Delta Y$ according to \eqref{eq:DeltaY}. In general, the distribution of $\Delta Y$ may be model--dependent, in that there is also a correlation between $\hat{\mu}_{r,i}$ and $\hat{\mu}_i$ in a given bin, and between the $(x_{r,i}-x_i)$ 
terms in different bins.  
Specifically, the perfect correlation between $x_{r,i}$ and $x_i$ leads to, 
in the case of a normal distribution for $M_i$,
\begin{equation}
    \Delta x_i \coloneqq x_{r,i}-x_i \sim N(0, f^2 \,(\hat{\mu}_{r,i}-\hat{\mu}_i)^2).
\label{eq:Deltaxi}
\end{equation}
In general, the variance of $\Delta x_i$ is approximately the same as in \eqref{eq:Deltaxi}, but when the 
distribution of $M_i$ is non--normal (e.g., a gamma distribution), different considerations 
must be used to obtain the distribution of $\Delta x_i$. Eq.~\ref{eq:Deltaxi} can therefore be used exactly 
when the $M_i$ are normal, or as an approximation in all other cases.

The following section considers a baseline constant model and a simple one--parameter extension, 
for which it is possible to find an analytic form for the  $\Delta Y$ statistic under the null hypothesis, and therefore
for the $\Delta C$ statistic in the presence of systematic errors. 
In general, the distribution of $\hat{\mu}_{r,i}-\hat{\mu}_i$, and therefore of $\Delta Y$, 
may  depend on model parameterization, and therefore additional considerations are required.
More general results
for the distribution of $\Delta Y$ that apply to any one--parameter extensions to a constant
reduced model are then presented in Sec.~\ref{sec:DeltaYq1}, and for the more general
case with $k \geq 1$ in Sec.~\ref{sec:extension}. A mathematical conjecture that justifies
these generalizations is discussed in 
App.~\ref{app:conjecture}.

\subsection{Constant model with a one--bin step--function}
\label{sec:DeltaYStep}
Consider, as an initial example, that the reduced model is a constant model with parent Poisson mean $\mu$ for 
all bins and thus with $m-k=1$; the full model is a constant model where
a fixed $j$--th bin is free to assume any value, therefore with $m=2$ free parameters (the overall
constant level and the level at the fixed position of the $j$--th bin) and $k=1$. This
can be considered as a toy model 
for the detection of unresolved features or fluctuations, and an approximation of models used in applications
\citep[i.e., the \texttt{line} model in \texttt{SPEX},][]{spence2023, bonamente2023}.

For this model, the $j$--th bin is the only bin where significant difference between the full and the
reduced model are expected, under common experimental conditions of extensive data ($N\gg 1$). 
In order to establish the asymptotic distribution of $\Delta Y$ for large $N$, it is 
therefore necessary to study the distribution of $(x_{r,j}-x_j)$ for the $j$--th bin where the additional
nested component is located.
Since the $N$ independent data points are distributed as $y_i \sim \mathrm{Poiss} (\mu)$, 
the fitted reduced constant model is the sample mean, and thus
$\hat{\mu}_{r,j} \sim N(\mu, \mu/N)$. On the other hand, the estimated mean for the full model at the
$j$--th position is 
$\hat{\mu}_j \sim \mathrm{Poiss} (\mu)$, since the full model has the flexibility to follow the $j$--th datum exactly due to the chosen form for the full model. In the large--mean limit, i.e., approximating the relevant
Poisson distributions with normal distributions of same mean and variance, 
it is thus approximately true that
\begin{equation}
    \hat{\mu}_{r,j}-\hat{\mu}_j \sim N(0,\mu(1+1/N) \simeq N(0,\mu).
\end{equation}

It therefore follows that the $\Delta Y$ distribution is a compounded normal distribution,
\begin{equation}
\begin{cases}
    \Delta Y \sim N(0, a^2\, \xi^2),\, \mathrm{with}\\
    \xi \coloneqq ( \hat{\mu}_{r,j}-\hat{\mu}_j) \sim  N(0,\mu) \text{ (for a fixed $j$)}
    \end{cases}
    \label{eq:compoundGauss}
\end{equation}
where the variance of $\xi$ is $\mu$, which is the fixed mean of the parent constant Poisson process, 
$f$ is the fixed value of the relative
systematic error, and $a\coloneqq2f$. 
Eq.~\ref{eq:compoundGauss} is the main result for the distribution
of the $\Delta Y$ statistic for this simple model with a one--bin step--function in a given $j$--th bin, and it
applies only in the large--mean limit. For data in the low--count regime, one cannot approximate a Poisson with 
a normal distribution, and therefore additional arguments would need to be used to find the distribution of $\Delta Y$ in this regime.

The compounded distribution of $\Delta Y$ according to \eqref{eq:compoundGauss} is said to be
a \emph{Bessel distribution} (of order zero),
\begin{equation}
\Delta Y \sim  K_0(\alpha)
\label{eq:Bessel}
\end{equation}
with scale parameter $\alpha= 2 f \sqrt{\mu}$ \citep[e.g.][]{mckay1932,craig1936,kotz2001}, where $K_0$
is the usual Bessel function of order 0.
The mean of
this distribution is zero, and 
the variance is
\begin{equation}
    \Var(\Delta Y) = 4\, f^2\, \mu,
\end{equation}
which is in fact consistent with the approximation of $\Var(\Delta Y) \simeq 4\, f^2\, y_j$ that
was previously provided in \cite{bonamente2023} using a simplified model of systematic errors. 
A feature of this distribution is a cusp in the pdf at $y=0$ of the Bessel function,
which is immediately seen as the
result of the normal distribution for the variance of $\Delta Y$.
Mathematical properties of this  distribution are described in more detail in \cite{bonamente2025properties}.

\subsection{Distribution of $\Delta Y$ for one additional parameter}
\label{sec:DeltaYq1}
The considerations and results provided above in Sect.~\ref{sec:DeltaYStep} can be generalized to 
other model parameterizations beyond the simple one--bin step function presented in the previous section.  
As an initial generalization, we 
continue with a reduced constant parent model,
and of a full model with just one additional nested parameter, such as the linear model. For this purpose, the following lemma and the associated theorem are
presented. 

\begin{lemma}[Distribution of sum of model deviations for a constant parent model and $k=1$]
\label{lemma1}
    Under the null hypothesis that the data are drawn from a (reduced) constant model
    with parent mean $\mu$, and that the
    full model has $k=1$ additional nested parameter, it is asymptotically true that
    \begin{equation}
        \sum_{i=1}^N (\hat{\mu}_{r,i}-\hat{\mu}_i) \sim N(0, \mu).
    \end{equation}
\end{lemma}
    \begin{proof}
        When the parent mean 
        is $\mu \gg 1$, it is possible to approximate
        \begin{equation}
            \Delta X \simeq \sum_{i=1}^N \dfrac{ \Delta \hat{\mu}_i^2}{\hat{\mu}_{r,i}}
            \label{eq:DeltaXApprox}
        \end{equation}
            where $\Delta \hat{\mu}_i = \hat{\mu}_{r,i}-\hat{\mu}_i$, and to approximate $y_i \simeq \hat{\mu}_{r,i}$,
            in accordance with the null hypothesis. When the model is constant, the approximation
            leads to 
        \begin{equation}
              \Delta X \simeq \dfrac{1}{\mu} \sum_{i=1}^N \Delta \hat{\mu}_i^2  
              \simeq  \dfrac{1}{\mu} \left(  \sum_{i=1}^N \Delta \hat{\mu}_i  \right)^2 \sim \chi^2(1),        
              \label{eq:DeltaXLemma1}
        \end{equation}
where the second approximation is due to the fact that,
to zeroth order,  $\sum \Delta \hat{\mu}_i \simeq 0$  according to
    the null hypothesis, and therefore the cross--product terms are negligible compared to the 
    $\Delta \hat{\mu}_i^2$ term.

        The distribution $\Delta X \sim \chi^2(1)$ applies under the null hypothesis, 
        providing the means to obtain a distribution 
        for its square root using the known fact that the square of a standard normal variable has a $\chi^2(1)$ distribution.
        While in general the converse is not necessarily true \citep[e.g.][]{roberts1966,roberts1971}, the square root of a $\chi^2(1)$ variable {is} in fact
        distributed as a standard normal under the assumption that the variable is symmetric \citep[e.g.][]{block1975}.
        This implies that, 
        within the limits of these assumptions and approximations,
        \begin{equation}
            \sum_{i=1}^N \Delta \hat{\mu}_i \sim N(0,\mu),
        \end{equation}
    where the variance of the normal distribution is $\mu$, according to \eqref{eq:DeltaXLemma1}.        

    \end{proof}


It is necessary to point out that the factorization of the parent mean \eqref{eq:DeltaXLemma1} is
required to
relate the $\Delta C$ statistic to the $\chi^2(1)$ distribution, and thus prove the lemma. Such factorization is only
possible when the parent model is constant. Therefore
Lemma~\ref{lemma1} is not guaranteed to apply in general, and additional considerations are required to establish an equivalent result when the reduced model is not constant. 

\begin{theorem}[Distribution of $\Delta Y$ for a constant parent model and $k=1$]
\label{th:DeltaY}
    Under the null hypothesis that the data are drawn from a reduced constant model
    with parent mean $\mu$, and that the
    full model has $k=1$ additional nested parameter, it is asymptotically true that
    \begin{equation}
        \Delta Y = 2 \sum_{i=1}^N (x_{r,i}-x_i) \sim K_0( 2\, f \sqrt{\mu}),
    \end{equation}
    i.e., $\Delta Y$ is distributed as a Bessel distribution with parameter 
    $\alpha= 2 \, f \sqrt{\mu}$, where $f$ is the value
    of the relative systematic errors as defined in \eqref{eq:f}.
\end{theorem}

\def\DeltamuiHat{\Delta \hat{\mu}_i}
\begin{proof}
    Following the same arguments as in Sect.~\ref{sec:DeltaYStep}, and
    with $\DeltamuiHat=\hat{\mu}_{r,i}-\hat{\mu}_i$ as previously defined,
    \[
    x_{r,i}-x_i \sim N(0, \DeltamuiHat^2\, f^2)
    \]
    due to the usual correlation between the two randomized values. Assuming independence in the
    bin--by--bin randomization, it is then true that
    \[
    \sum_{i=1}^N (x_{r,i}-x_i) \sim N\left(0, \sum_{i=1}^N \DeltamuiHat^2\, f^2\right).
    \]
    The
    $\Delta Y$ variable is compounded according to an equivalent relationship to \eqref{eq:compoundGauss},
    \begin{equation}
\begin{cases}
    \Delta Y \sim N(0, a^2\, \xi^2),\, \mathrm{with}\\
    \xi= \sum_{i=1}^N \DeltamuiHat \sim N(0,\mu),
    \end{cases}
    \label{eq:compoundGaussLinear}
\end{equation}
with the usual $a=2\, f$, and with  Lemma~\ref{lemma1} being used for the distribution of the sum of model
differences. The rest of the theorem is 
due to the definition of a Bessel distribution according to \eqref{eq:compoundGauss}.
\end{proof}

Theorem~\ref{th:DeltaY} therefore establishes that, in the limit of a large parent Poisson mean
and under the null hypothesis of a constant reduced model,
the $\Delta Y$ statistic has a Bessel distribution.
This result is expected to hold for any model with one
additional parameter (e.g., the linear model) relative to the baseline  constant model.


\subsection{Extension to other parameterizations and multiple parameters in the nested component}
\label{sec:extension}
It is useful to seek an extension of the results of Sec.~\ref{sec:DeltaYq1} to any reduced model beyond the simple
constant model,
and for any number of parameters for the nested model component. This section discusses this general situation, which is necessary in order to use systematic errors for the \DeltaC\ statistics
in practical applications that typically feature more complex models, and often $k>1$ parameters
in the nested component.
 In the following we present practical considerations for the
use of this model of systematic errors in the general case, and in 
App.~\ref{app:conjecture} we
 outline a path towards a mathematical proof of these considerations.

Starting with the results of a Bessel--distributed $\Delta Y \sim K_0(\alpha)$ with $\alpha=2\,f \sqrt{\mu}$ from
Th.~\ref{th:DeltaY}, which applies exactly only for a constant model and one additional parameter ($k=1$), 
it appears reasonable to entertain a generalization that features 
\begin{equation}
    \alpha_k=2\,f \sqrt{k\cdot \overline{\mu}}
    \label{eq:alphaPractical}
\end{equation}
when $k\geq 1$,  where  $\overline{\mu}$
represents a suitable average of the parent Poisson mean $\mu$ over the range of the
independent variable $x$.
In fact, the variance of the Bessel distribution is $\alpha^2$ \citep[e.g.][]{kotz2001,bonamente2025properties}, and such generalization would be the result of the linear addition of $k$ independent 
contributions. This is the practical sense behind Conjecture~\ref{conjectureDeltaX} and its associated 
corollaries that are proposed in a more formal way in App.~\ref{app:conjecture}. Therefore it is reasonable to expect that
\begin{equation}
    \Delta Y \sim K_0(\alpha_k),
    \label{eq:DeltaYK0}
\end{equation}
which is the same distribution as in Corollary~\ref{th:corollary2} discussed
in App.~\ref{app:conjecture}.

Alternatively, we entertain the possibility that the distribution of the 
$\Delta Y$ statistic is the sum of $k$ Bessel distributions, each of the same type discussed in the previous 
paragraph for the case of $k=1$. In this case, and under the additional assumption of independence among these 
(random variable) contributions, the $\Delta Y$ statistic will have the same zero mean
and variance as given by \eqref{eq:alphaPractical} or \eqref{eq:alphaq}, 
but be asymptotically
\emph{normally} distributed according to the central limit theorem when $k$ becomes large. 
In this case, it would be asymptotically true that
\begin{equation}
    \Delta Y \overset{a}{\sim} N(0,\alpha_k)
    \label{eq:DeltaYPractical}
\end{equation} 
when there are several
adjustable parameters ($k\gg 1$) in the nested component. This possibility is in fact consistent with 
an earlier 
model of systematic errors for $\Delta C$  \citep[e.g., Eq.~26 of][]{bonamente2023}.

The two distributions for $\Delta Y$ proposed in this section will be tested with the aid of numerical
simulations in the following section.  The reader is referred to App.~\ref{app:conjecture} for a more
formal treatment that leads to \eqref{eq:DeltaYK0}.

\subsection{Numerical tests of the $\Delta Y$ distribution}

We performed a series of numerical simulations to test the distributions of $\Delta Y$ that have
been 
proposed in this section. The numerical simulations follow the same
method as those presented in \paperOne, with $N=100$ bins and a constant model used as the
baseline or reduced model, and a Monte Carlo simulation with 1,000 iterations.

\begin{figure*}
    \centering
    \includegraphics[width=3.4in]{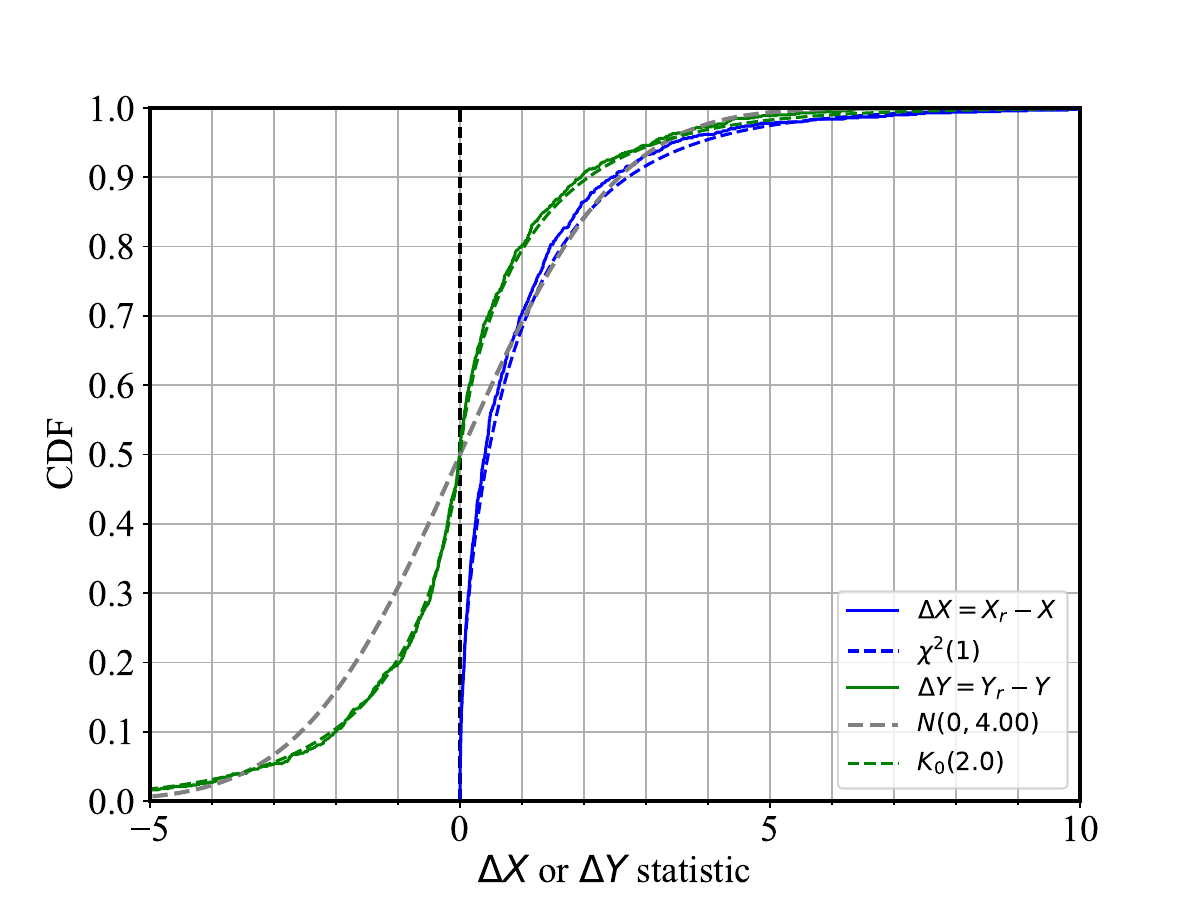}
    \includegraphics[width=3.5in]{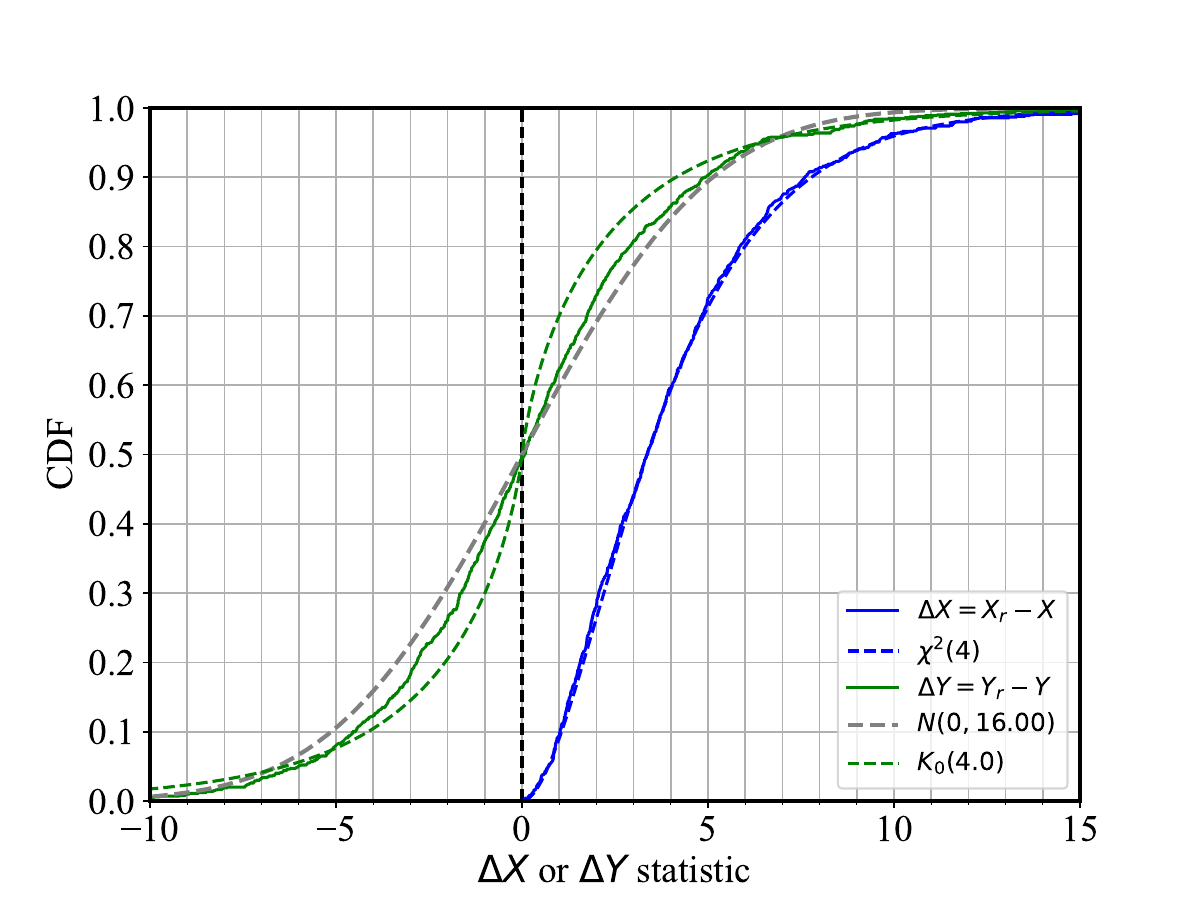}
    \caption{(Left:) Cumulative distribution functions for the $\Delta X$ and $\Delta Y$ statistics that apply to the Poisson regression with the one-bin step-function, for a 10\% systematic error. 
    (Right:) Same, but for the $k$--bin step--function for $k=4$. Overplotted are the
    chi--squared distributions for the $\Delta X$ statistics, and the Bessel and normal approximations
    for the $\Delta Y$ statistics according to \eqref{eq:DeltaYK0} and \eqref{eq:DeltaYPractical}.}
    \label{fig:DeltaxDeltaY}
\end{figure*}

\subsubsection{Tests for $k=1$ (one--parameter nested component)}
\label{sec:DeltaXDeltaY}
First, the full model consists of the one--bin step--function modification to the baseline constant model,
as described in Sect.~\ref{sec:DeltaYStep}, with a 10\% systematic error ($f=0.1$). 
At each iteration, a Poisson dataset is drawn from the parent distribution with a constant mean $\mu$, and the
data are fit to both the full and the reduced model. 
The empirical CDF for the relevant \cmin\ statistics were illustrated and discussed
in \paperOne, with the general result that the $\cmineq=X+Y$ statistics are accurately described by the 
normal distribution \eqref{eq:ZAsymptotic} that is applicable in the large--mean and extensive data regime of this simulations. These distributions are not shown in this paper; instead, we focus
on the $\Delta X$ and $\Delta Y$ statistics as defined in \eqref{eq:DeltaCsys}.

After the best--fit models $\hat{\mu}_i$ and $\hat{\mu}_{r,i}$ are obtained, they are randomized following the method described in 
Sec.~\ref{sec:DeltaYGeneral}, using $f=0.1$ or a 10\% level of systematic errors. 
The left panel of Figure~\ref{fig:DeltaxDeltaY} shows the resulting eCDFs for the statistic.  In blue is the $\Delta X$ statistic, which as expected follows closely a $\chi^2(1)$ distribution.
The $\Delta Y$ distribution is illustrated as a green solid curve, and it follows closely
the expected Bessel distribution (dashed green curve) with an expected variance of $\Var(\Delta Y)=4$, for a choice of $f=0.1$ and $\mu=100$. This simulation therefore illustrates that the compounded distribution for the $x_{r,i}-x_i$ variable described in Sect.~\ref{sec:DeltaYGeneral}  can be successfully used to model the $\Delta Y$ contribution to the $\Delta C$ statistic induced by systematic errors, for the simple one--bin step--function modification of the constant model. For comparison, a normal distribution of same variance 
(and same null mean) is overplotted as a red dashed curve, to show that it is a poorer match to the simulated eCDF of the statistic.

Next, the same simulation is repeated by using the linear model 
as the full model, in place of the one--bin step--function modification to the baseline constant model, as the
full model. The results of this Monte Carlo simulation 
are virtually identical to those of Fig.~\ref{fig:DeltaxDeltaY} and therefore they are not shown, with
the $\Delta Y$ variable closely following the appropriate Bessel distribution.
This agreement follows the results of Sec.~\ref{sec:DeltaYq1}, where the Bessel $B(\sigma)$ distribution
was shown to apply to the $\Delta Y$ statistic for any one--parameter nested component added to a constant
model.

\subsubsection{Tests for $k>1$ (multi--parameter nested component)}
For this purpose, we first consider a model similar to the one--bin step function considered above,
except that the model now has $k$ step functions instead of one, for a total of $k$ model parameters in the nested component 
corresponding to the normalizations in those indepedent bins.~\footnote{The location
of these bins in the $x$ variable is irrelevant, given the independence of the Poisson data.}
According to \eqref{eq:DeltaYK0}, the distribution of $\Delta Y$ is expected to follow a Bessel distribution with parameter $\alpha_k=2 f \sqrt{k \mu}$.
The simulations performed with the $k$--bin step--function model show that this model
is accurate when $k$ is a small number (e.g., $k=2$ or 3), and $\Delta Y$ progressively
tends to a normal distribution of same variance according to \eqref{eq:DeltaYPractical} 
when $k$ becomes larger. A representative situation is illustrated in the right--hand panel of
Fig.~\ref{fig:DeltaxDeltaY} for the case of $k=4$, where the simulated eCDF is between expected the 
Bessel and normal distributions, consistent with the discussion of Sec.~\ref{sec:extension}.
 For the Monte Carlo simulation of Fig.~\ref{fig:DeltaxDeltaY}, 
the same $\mu=100$, $f=0.1$ and $N=100$ parameters as in Sec.~\ref{sec:DeltaYStep} were used. 
Similar results apply to other choices of the parameters $\mu$ and $f$, and they are not reported in the paper.

Next, we consider a polynomial model of order $k$, e.g., $y=a_o+a_1\,x+ \dots +a_k\,x^k$, which generalizes the constant model with a nested model component with $k$ parameters $a_k$. In particular, we performed a series of numerical simulations for $k=1$ and $k=2$, which 
yield quantitatively similar results to those illustrated in Figs.~\ref{fig:DeltaxDeltaY}: the case of $k=1$
follows the expected theoretical behavior of a Bessel function, and the case of $k=2$ shows only small
deviations towards normality, in both cases with the expected mean and variance. Those
tests lend additional support for the results presented in this paper. 

We also notice that, for larger
values of $k$, the polynomial model appears to suffer from the problem of parameter \emph{unidentifiability}, 
which is manifested
as the $\Delta X$ distribution itself not following the chi--squared distribution as prescribed by the 
Wilks theorem. This is a fundamental problem 
of statistical estimation  \citep[e.g., Sec.~29.11, ][]{kendall1979} that has received much attention
in econometrics \citep[e.g.,][]{wald1950,fisher1966,amemiya1985} and other disciplines  \citep[e.g.,][]{raue2009,GODFREY1985}.
Higher--order polynomial terms suffer from this problem, especially for the 
type of noisy data under consideration, whereas the $k$--bin step--function model does not, since
the latter relies on a specific datum to determine the associated model parameter, whereas the former
would be required to estimate parameters ($a_k$ for large $k$) that the data 
are simply unable to determine (see also \citealt{bonamente2024properties} for further discussion on parameter
identifiability).

\subsubsection{Correlation between $\Delta X$ and $\Delta Y$}
\label{sec:correlation}
Numerical simulations also provide empirical estimates of the correlation between
the $\Delta X$ and $\Delta Y$ statistics. For this purpose, we performed
six sets of ten simulations of the type shown in Fig.~\ref{fig:DeltaxDeltaY}, respectively with $f=0.01$
and $f=0.10$, representative of a small and large value of the
systematic error; and for $k=1,3$ and 5--bin step functions.
All simulations have $N=100$ datapoints and 
a $\mu=100$ parent mean, same as in Fig.~\ref{fig:DeltaxDeltaY}.

The mean sample correlation coefficient between $\Delta X$ and $\Delta Y$ 
for the 60 simulations was $r=-0.0165\pm0.0071$ (sample and standard deviation of the mean), with a standard deviation  
of $0.0547$ among all the simulations. The $f=0.1$ simulations (10\% systematic error) had sample correlation
coefficients of $r=-0.0031\pm0.0837$, $-0.0315\pm0.0433$ and $-0.0066\pm0.0349$
respectively for $k=1,3,5$; and the $f=0.01$ simulations (1\% systematic error)
values of $r=-0.0179\pm0.0698$, $-0.0250\pm0.0408$ and $-0.0151\pm0.0483$.
These simulations provide indication that the two statistics are nearly uncorrelated,
at most with a percent--level amount of correlation that appears to be preferentially negative.

\subsubsection{Summary of numerical tests}
The overall success of the Monte Carlo simulations with $k \geq 1$ lend support to the applicability
of the Bessel distribution for $\Delta Y$, and asymptotically of a normal distribution when $k$ is large, according to the results of Sec.~\ref{sec:extension}. In particular, Eq.~\ref{eq:DeltaYK0} and
\ref{eq:alphaPractical} requires the identification of a `suitable'
average $\overline{\mu}$, which in this application 
we successfully set
to the parent mean of the reduced model, i.e., $\overline{\mu}=\mu$. It is therefore reasonable
to speculate that, for more complex models beyond the constant, a similar average can be found.

With regards to the independence between $\Delta X$ and $\Delta Y$,
the tests of Sec.~\ref{sec:correlation} suggest that the two statistics $\Delta X$ and $\Delta Y$ are nearly uncorrelated, likely with a small degree of negative correlation. Since uncorrelation is only a necessary (but not sufficient) condition for independence, dependence between the two variables is still possible. Such dependence will be examined further is the following section.

\section{Distribution  of the \DeltaCSys\ statistic}
\label{sec:DeltaC}
We are now in a position to turn to the overall distribution of the \DeltaCSys\ statistic
in the presence of systematic errors, as defined
in \eqref{eq:DeltaCsys}, for the general case of $m \geq 1$ free
parameters with $1 \leq k\leq m$ free parameters in a nested component. 
 The distribution of the $\Delta C=\Delta X + \Delta Y$ statistic can be obtained
 as the convolution of the two  distributions, assuming 
 independence.

For the \cmin\ \gof\ statistic, 
independence between the individual components $X\sim \chi^2(N-m)$ and $Y \sim N(\hat{\mu}_B,\hat{\sigma}^2_C)$ for the full model and the reduced model
follows from the argument presented in \paperOne\ and summarized in Sec.~\ref{sec:paperOneSummary}. 
For the $\Delta C$ statistic, however, independence between the two
contributing statistics is not guaranteed. In fact, $\Delta X$
is a function of $\Delta \hat{\mu}_i$ according to \eqref{eq:DeltaXApprox}, and $\Delta Y$ according to \eqref{eq:DeltaY} is a function
of $x_{r,i}-x_i \sim N(0, f^2\,\Delta \hat{\mu}_i^2)$, thus a degree of correlation
between $\Delta X$ and $\Delta Y$ may be present, 
which was in fact investigated in Sec.~\ref{sec:correlation}. 

Two alternatives are proposed in order to use the present model of systematic errors to determine the distribution
of $\Delta C$ and therefore enable a quantitative hypothesis testing method: an exact method based on Monte Carlo simulations, and an approximate method that ignores the possible dependence between
$\Delta X$ and $\Delta Y$. The two methods are discussed in the following and tested with numerical simulations.

\subsection{Distribution of \DeltaCSys\ via Monte Carlo simulations}
\label{sec:MonteCarlo}
The most accurate method to
 determine the distribution of $\Delta C$ is via a Monte Carlo simulation that can be summarized
 in the following steps:\\
 1. Generation of a Poisson dataset, drawn from the reduced model that corresponds to the null hypothesis. In this
 paper the constant model was used, but any model can be used, according to the application
 at hand.\\
 2. Regression of the data with the baseline model, leading to the $\hat{\mu}_{r,i}$ best--fit means for
 each of the $N$ bins. No systematic errors are used for this regression.\\
 3. Regression with the full model leading to the $\hat{\mu}_i$ best--fit means, again with no
 systematic errors.\\
 4. Randomization of the best--fit models $\hat{\mu}_{r,i}$  and $\hat{\mu}_i$ according to \eqref{eq:mui}, so that randomized values of the best fit models ($\mu_{r,i}=\hat{\mu}_{r,i}+x_{r,i}$  and $\mu_i=\hat{\mu}_i+x_{i}$) are obtained.\\
 5. Calculation of  the \cmin\ statistic for both the full and the reduced model, using the 
 randomized best--fit models, and calculations of the $\Delta C$ statistic per \eqref{eq:DeltaCsys}.\\
 6. Iteration of steps 1--5 to obtain a large number of Monte Carlo samples for the empirical distribution
 of $\Delta C$.
 
This is the method that was used for the Monte Carlo simulations leading to the eCDF (black curves) of Fig.~\ref{fig:DeltaC}, where the chi--square distributions that apply to the case of no systematic errors are illustrated as blue dashed curves. The eCDF generated in this manner can then be used to determine critical values at any level of confidence for the purpose of hypothesis testing. 

\begin{figure*}
    \centering
    \includegraphics[width=3.5in]{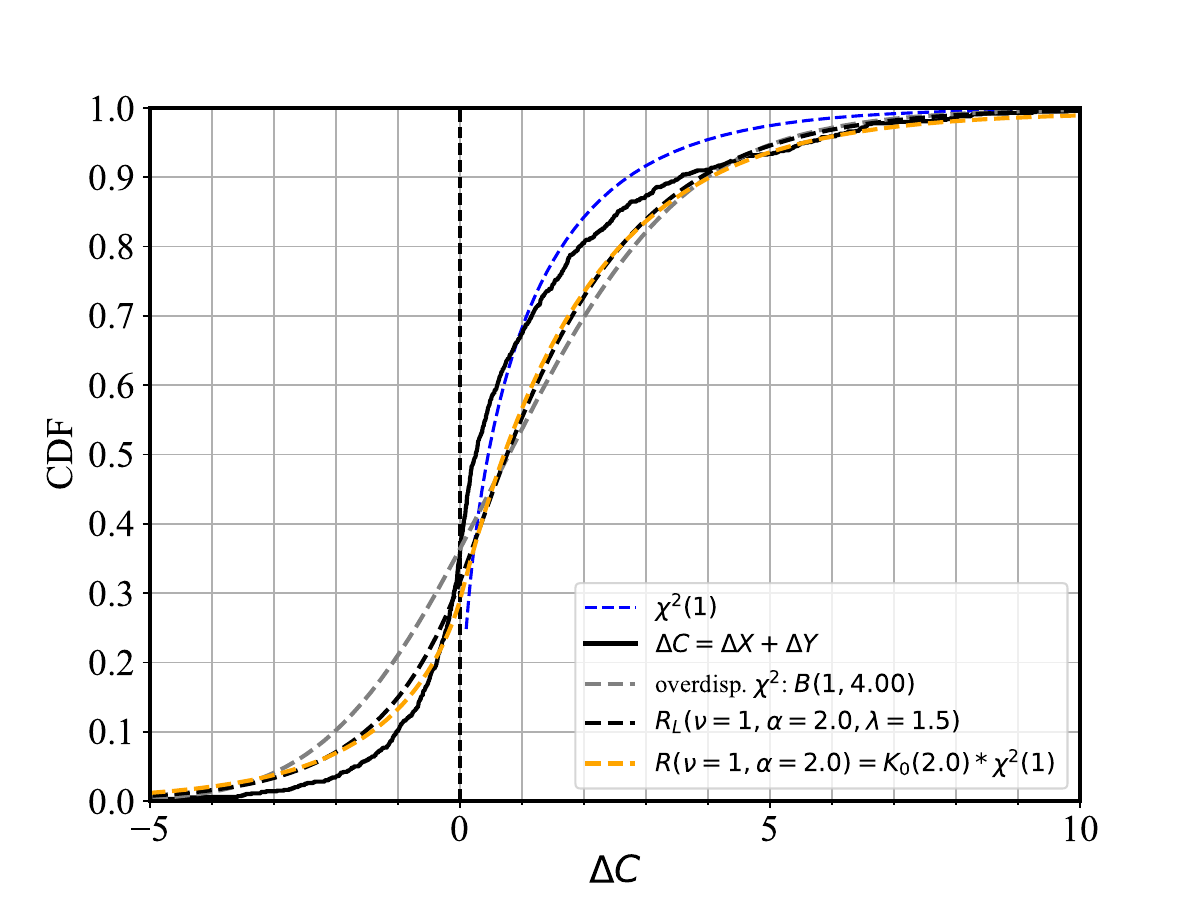}
        \includegraphics[width=3.5in]{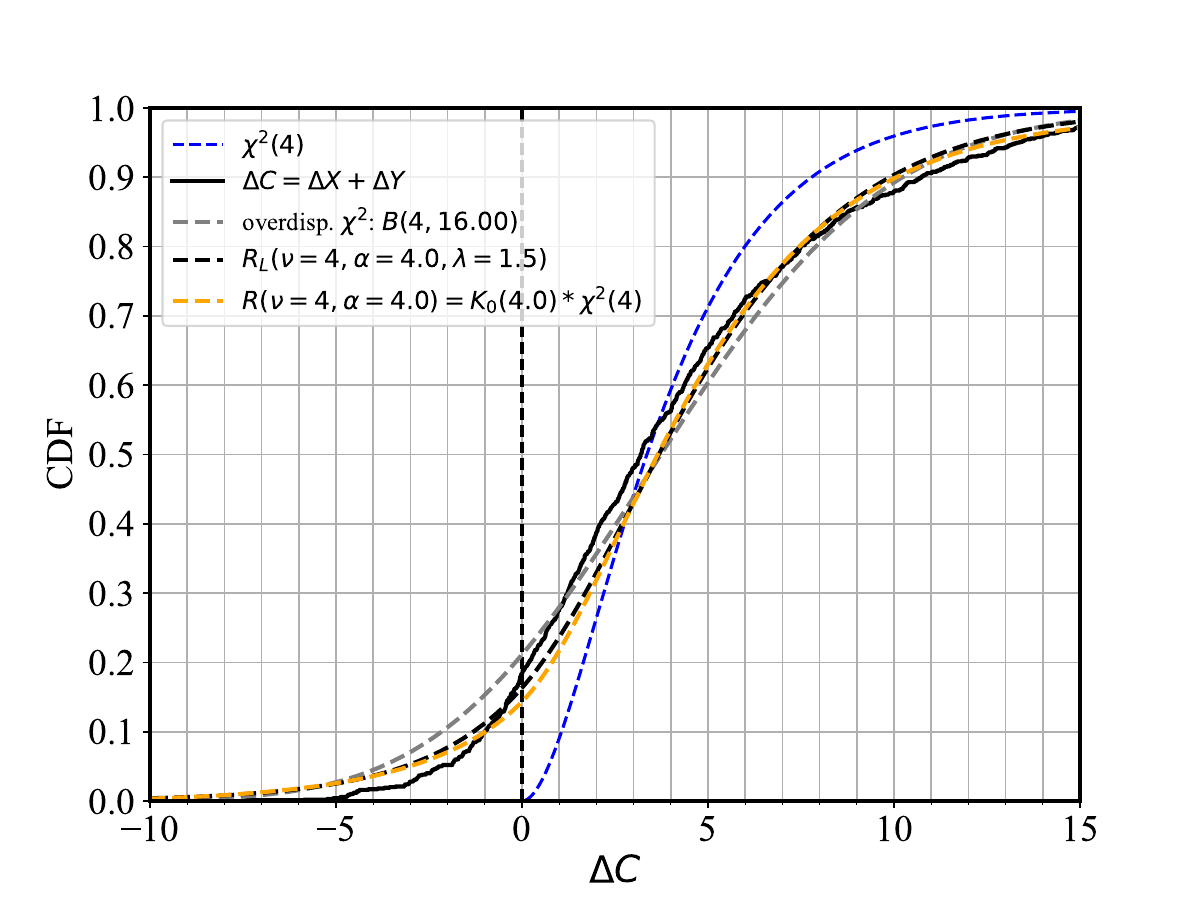}
    \caption{Cumulative distribution functions for the $\Delta C$ statistic using the same models
    as in Fig.~\ref{fig:DeltaxDeltaY}. Overplotted are the $\chi^2(1)$ distribution that applies to the case of
    no systematic errors, and  distributions used as a possible theoretical model: 
    the randomized $\chi^2$ distribution $R_L(\nu,\alpha,\lambda)$, with the Laplace approximation to the Bessel distribution; the  randomized chi--squared distribution $R(\nu,\alpha)$ obtained by direct convolution with
    the Bessel distribution; and the overdispersed chi--squared distribution, in which the Bessel
    distribution was approximated with a normal distribution prior to the convolution with the chi--squared distribution \citep{bonamente2023}.}
    \label{fig:DeltaC}
\end{figure*}

\subsection{Approximate analytic distributions of \DeltaCSys}
\label{sec:randomizedChi}

It is useful to investigate the possibility of finding  an approximate 
analytic form for the distribution for the \DeltaCSys\
statistic, in order to overcome the need for a Monte Carlo simulation for any given application. To this end, it is assumed that
 the probability distribution  of the $\Delta C$ statistic is obtained from the convolution integral
 of the two densities, assuming independence between
 $\Delta X$ and $\Delta Y$.  As discussed earlier in the section, one 
 may expect a degree of dependence 
 between $\Delta X$ and $\Delta C$, and therefore the method described in this section should  
 be considered an approximation. 
 
\subsubsection{The randomized $\chi^2$ distribution $R(\nu,\alpha)$}
Consider two independent random variables $X\sim \chi^2(\nu)$ and $Y \sim K_0(\alpha)$, the latter a Bessel
distribution with parameter $\alpha$ \citep{bonamente2025properties}. The convolution of the
two probability densities
simplifies to 
 \begin{equation}
     f_{R}(z)=\int_{0}^{\infty} f_{K_0}(z-y; \alpha) f_{\chi^2}(y; \nu)\, dy
     \label{eq:fDeltaC}
 \end{equation}
where the positive--value constraint of the $\chi^2(\nu)$ distribution is enforced. 
The density of a Bessel distribution with parameter $\alpha$
is given by 
\begin{equation}
f_{K_0}(x; \alpha) = \dfrac{1}{\pi \alpha} K_0\left(\left| \dfrac{x}{\alpha}\right| \right) 
\label{eq:fK0}
\end{equation}
with $\alpha$ the parameter of the (zero order) Bessel distribution, and $K_0(x)$ the modified Bessel function of order zero \citep[e.g.][]{kotz2001,bonamente2025properties}. 

The  family of distributions
that follow the pdf of \eqref{eq:fDeltaC},  i.e., the convolution of a chi--squared and a normal distribution, will be referred to as the family of 
\emph{randomized} $\chi^2$ distributions $R(\nu,\alpha)$  with parameters $\nu \in \mathbb{N}$ and $\alpha \geq 0$ a real number.~\footnote{This distribution applies also when $\nu \in \mathbb{R}$,
although for this class of applications it is only interesting to consider the case of natural numbers.} This family of distributions are the approximate parent distribution for \DeltaCSys, under the simplifying assumption of independence between the two contributing statistics $\Delta X$ and $ \Delta Y$. 
The distribution in its general form must be evaluated numerically, and it is shown as  dashed orange 
curves in Fig.~\ref{fig:DeltaC}.

\subsubsection{Analytic approximation  $R_L(\nu,\alpha, \lambda)$}
It is useful to seek a simple
approximation to the Bessel distribution \eqref{eq:fK0},  which is part of the
integrand in  \eqref{eq:fDeltaC}, to obtain a simple closed form for the
convolution integral of the randomized chi--squared distribution. An avenue is
provided by the series expansion provided by, e.g., \cite{martin2022}, 
whereby retaining the zero-th order term of the expansion provides a good approximation to the distribution, as shown in \cite{bonamente2025properties}. That approximation, however, has a singularity at $x=0$ that makes it
unsuitable for the task.
Another approximation that circumvents this problem is a zero--mean Laplace 
(also known as double--exponential) distribution with density of the type
\begin{equation}
    f_{\text{L}}(x; \alpha, \lambda) = \dfrac{\lambda}{2 \alpha}e^{-\dfrac{\lambda |x|}{\alpha}}.
    \label{eq:fLaplace}
\end{equation}
 This distribution will be referred to as $\Laplace(\alpha/\lambda)$, i.e., a zero--mean Laplace distribution with parameter $\alpha/\lambda$, with $\lambda$ a fixed constant that serves the purpose to provide a good fit to the Bessel distribution, as discussed in \cite{bonamente2025properties}, and $\alpha$ the same
 parameter as defined for the Bessel distribution. ~\footnote{It is clear that the distribution \eqref{eq:fLaplace} depends only on the ratio $\alpha/\lambda$. However, it is convenient to retain both constants for ease of interpretation.}
By approximating the Bessel distribution with a Laplace distribution of suitable parameter $\lambda$, i.e., 
$f_{K_0}(x; \alpha) \simeq f_{\text{L}}(x; \alpha, \lambda)$,
it is possible to find an approximate analytic form for the density of  \DeltaCSys\ 
(see App.~\ref{app:Laplace} for details). This family of distributions is referred to a $R_L(\nu,\alpha,\lambda)$ with density $f_{R_L}$ reported in \eqref{eq:randChiLaplace}.
 
 The randomized $\chi^2$ distribution $ R(\nu,\alpha)$, shown as dashed black curve in Fig.~\ref{fig:DeltaC}, has the same two key features
 that were observed in the distribution of \DeltaCSys\ from Monte Carlo simulations, namely
 a tail of negative values and a wider right tail, compared to the $\chi^2$ distribution.
It is found that a value of $\lambda=1.5$ provides a reasonable approximation so
that $R_L(\nu,\alpha,\lambda=1.5)\simeq R(\nu,\alpha)$, as shown in App.~\ref{app:Laplace} (see also \citealt{bonamente2025properties}).

\subsubsection{The overdispersed chi--squared distribution $B(\nu,\mu=0,\alpha)$}
\label{sec:overdispersed}
If the $\Delta Y$ statistic is approximated by a normal distribution, which appears to be a suitable
option when $k$ is large (see Sec.~\ref{sec:extension}), then the convolution between a chi--squared distribution $\chi^2(\nu)$ and
a normal $N(0,\alpha)$ leads to an overdispersed chi--squared distribution with zero mean, $B(\nu,\mu=0,\alpha)$. 
This family of distributions has a closed form for its density, and it was described extensively in 
\paperOne\ and \cite{bonamente2024properties}. It is shown as a dashed grey curve in Fig.~\ref{fig:DeltaC} to illustrate
the difference between the Bessel and normal approximations for the distribution of $\Delta Y$, and their effect on \DeltaCSys.

\subsection{Comparison of \DeltaCSys\ distributions}

Fig.~\ref{fig:DeltaC} provides a
 comparison between the eCDFs  obtained from Monte Carlo simulations (solid black curves) and the CDF of the 
 $R(\nu,\alpha)$  randomized $\chi^2$ distribution (orange) and the 
 $R_L(\nu,\alpha,\lambda=1.5)$ 
distribution (dashed black curves) that uses the Laplace approximation to the Bessel distribution,
for two representative cases with $\nu=k=1$ and $\nu=k=4$, for the same simulations as in Fig.~\ref{fig:DeltaxDeltaY}. 
Comparison between these curves can be used to assess the goodness of the hypothesis that were used for
their calculation. Similar Monte Carlo simulations for a range of parent Poisson means $\mu$ and level
of systematic errors $f \ll 1$ were also performed, confirming the general features present in Fig.~\ref{fig:DeltaC}. 
For comparison, the overdispersed chi--squared distribution (where the Bessel distribution is replaced by a normal
distribution, see Sec.~\ref{sec:overdispersed}) is also shown as a dashed grey curve.

Comparison among the distributions reveals the following general features:\\
(a) There is an excellent agreement among the experimental eCDF, the $R(\nu,\alpha)$ and the
$R_L(\nu,\alpha,\lambda=1.5)$ distributions on the value of the $z=0$ quantile, i.e., the CDFs overlap at $z=0$.
This indicates that the probability of a negative value for the statistic can be equally well be estimated by
either of the randomized $\chi^2$ distributions obtained by the convolution.

(b) There appear to be systematic differences between the eCDF and the randomized $\chi^2$ distribution near $z=0$,
i.e., the eCDF is systematically lower for small negative values, and systematically
larger
for small positive values. This is likely attributable to the correlation between $\Delta X$ and $\Delta Y$
that was discussed in Sec.~\ref{sec:correlation}, and that was ignored in the convolution. These systematic 
difference would result in erroneous estimates of $\Delta C$ quantiles at approximately $p \leq 0.9$ or so.

(c) There is good agreement among the distributions in the right tail, making it
possible to use the $R_L(\nu,\alpha)$ or the $R(\nu,\alpha)$ distributions to estimate 
quantiles 
for probabilities $p\geq 0.9$, which is the most common task for
the type of one--sided hypothesis tests for the significance of a nested component.
This feature is very convenient for the use of the analytical approximation $R_L(\nu,\alpha)$  to estimate
critical values of the $\Delta C$ statistic for small null--hypothesis probabilities, i.e., $1-p \leq 0.1$, i.e., in the right tail of the distribution. 

Similar features are displayed by other simulations for different values of the parent mean $\mu$, the level
of systematic errors $f$, and the number of additional parameters $k$. 
The numerical simulations discussed in this section indicate that the most accurate method to
determine the parent distribution for $\Delta C$ in a given application is to perform a Monte Carlo
simulation of the type discussed in Sec.~\ref{sec:MonteCarlo}. Moreover, the success of the
randomized chi--square distributions $R(\nu,\alpha)$ and $R_L(\nu,\alpha,\lambda)$, and of the 
overdispersed chi--squared distribution $B(\nu,\mu=0,\alpha)$ for $k>1$, in reproducing the right tail of the distribution
of \DeltaCSys\
suggest that these approximations can be used to estimate large quantiles ($p\geq 0.9$) of the distribution.
Since the typical hypothesis testing task for the significance of a nested model component
is to estimate these one--sided critical values (with small $1-p$), we conclude that these distribution
can be used with good accuracy for hypothesis testing. More extensive
numerical tests go beyond the scope of this paper, and are deferred to a separate paper.

\section{Applications to hypothesis testing}
\label{sec:HypothesisTesting}

\subsection{Methods of hypothesis testing with \DeltaCSys}
\label{sec:DeltaCTestMethod}

Hypothesis testing for the presence of a nested model component 
consists of comparing the measured $\Delta C$ value with critical values of a parent distribution,
at a given confidence level. In the absence of systematic errors it is expected that $\Delta C \sim \chi^2(k)$,
where $k$ is the number of parameters on the nested component,
with certain restrictions that require the null--hypothesis values of the additional parameter to lie in the 
interior of 
the allowable parameter space (or range), and not at its boundaries \citep[e.g.][]{protassov2002}. 

This model of systematic errors proposes that the \DeltaCSys\ statistics follows approximately
a randomized chi--squared distribution,
$\DeltaCSysEq \sim R(\nu,\alpha_k)$. The two parameters of the
distribution are respectively 
$\nu=k$, representing the number of parameters in the nested component; and  $\alpha$ is a function of the amount of systematic errors according to
\eqref{eq:alphaPractical}. In this paper we have proven this result for the case $k=1$, i.e., for
a nested model component with one additional parameter, and for a constant baseline model. We have
also conjectured that these results may also apply for $k \geq 1$ and for any parameterization
of the baseline and full models (see Sec.~\ref{sec:extension}). We have also provided 
a mathematical conjecture that lays out a possible path towards an exact proof of these results (see App.~\ref{app:conjecture}). 

 Selected critical values according to the
$R(\nu,\alpha)$ distribution, and its approximation $R_L(\nu,\alpha,\lambda)$, are presented in
Table~\ref{tab:critVal}. The parameter $\alpha_k$  combines the amount of systematic errors
($f \ll 1$) with the parent Poisson mean $\mu$ and the number of parameters $k$. For example, $\alpha=2$ may result from a parent mean $\mu=100$
in the presence of a 10\% ($f=0.1$) level of systematic errors for $k=1$, or any other combination that results in the
same product. This model of systematic errors is valid for values of approximately $f \leq 0.1$ or so, as
discussed in \paperOne, and cannot be used for values of $f$ close to one. 

The critical values of Table~\ref{tab:critVal} show that a few--percent level of systematic errors
in a large--mean Poisson dataset has a significant effect on critical values for the $\Delta C$ statistic. For example, $\alpha=2$ for one additional parameter 
results in a $q=1-p=0.10$ critical
value (or $1-p=90$\% confidence) of 3.9, versus the value of 2.7 that applies to the standard chi--squared distribution, i.e., an increase by $\geq 40$~\%. As expected, the effect of systematic errors is that
of reducing the power
of detection of a model component, compared to the parent $\chi^2(k)$
distribution that applies when there are no systematic errors.

\begin{table*}[!t]
    \caption{One--sided critical values for a randomized chi--squared distribution $R_L(\nu,\alpha,\lambda)$ with $\lambda=1.5$ the index of the Laplace distribution that approximates the Bessel distribution $K(\alpha)$.  In parenthesis are 
    reported critical values of a $\chi^2(\nu)$ distribution, which correspond to the case of no systematic errors. Probabilities $q=1-p$ correspond to one--sided critical values $x=F_{R_L}^{-1}(p)$.}
    \label{tab:critVal}
    \centering
    \normalsize
     \begin{tabular}{llllll}
    \hline
    \hline
     & \multicolumn{5}{c}{Critical values $R_L(\nu,\alpha,\lambda=1.5)$}\\
    $\alpha$ & $\nu=1$ & $\nu=2$ & $\nu=3$ & $\nu=4$ & $\nu=5$ \\
\hline
\multicolumn{6}{c}{q=0.317}\\
1.0 & 1.3 (1.0)& 2.5 (2.3)& 3.7 (3.5)& 4.8 (4.7)& 6.0 (5.9)\\
2.0 & 1.7 (1.0)& 2.8 (2.3)& 4.0 (3.5)& 5.1 (4.7)& 6.2 (5.9)\\
3.0 & 2.0 (1.0)& 3.2 (2.3)& 4.3 (3.5)& 5.4 (4.7)& 6.5 (5.9)\\
4.0 & 2.4 (1.0)& 3.5 (2.3)& 4.6 (3.5)& 5.7 (4.7)& 6.8 (5.9)\\
5.0 & 2.7 (1.0)& 3.8 (2.3)& 4.9 (3.5)& 6.0 (4.7)& 7.1 (5.9)\\
\hline
\multicolumn{6}{c}{q=0.100}\\
1.0 & 3.1 (2.7)& 4.8 (4.6)& 6.4 (6.3)& 7.9 (7.8)& 9.4 (9.2)\\
2.0 & 3.9 (2.7)& 5.5 (4.6)& 7.0 (6.3)& 8.4 (7.8)& 9.8 (9.2)\\
3.0 & 4.8 (2.7)& 6.3 (4.6)& 7.7 (6.3)& 9.1 (7.8)& 10.5 (9.2)\\
4.0 & 5.8 (2.7)& 7.2 (4.6)& 8.5 (6.3)& 9.9 (7.8)& 11.2 (9.2)\\
5.0 & 6.8 (2.7)& 8.1 (4.6)& 9.4 (6.3)& 10.7 (7.8)& 12.0 (9.2)\\
\hline
\multicolumn{6}{c}{q=0.050}\\
1.0 & 4.2 (3.8)& 6.2 (6.0)& 8.0 (7.8)& 9.7 (9.5)& 11.2 (11.1)\\
2.0 & 5.1 (3.8)& 7.0 (6.0)& 8.6 (7.8)& 10.2 (9.5)& 11.7 (11.1)\\
3.0 & 6.4 (3.8)& 8.0 (6.0)& 9.6 (7.8)& 11.1 (9.5)& 12.5 (11.1)\\
4.0 & 7.7 (3.8)& 9.2 (6.0)& 10.7 (7.8)& 12.1 (9.5)& 13.5 (11.1)\\
5.0 & 9.1 (3.8)& 10.5 (6.0)& 11.9 (7.8)& 13.3 (9.5)& 14.6 (11.1)\\
\hline
\multicolumn{6}{c}{q=0.010}\\
1.0 & 6.9 (6.6)& 9.4 (9.2)& 11.6 (11.3)& 13.5 (13.3)& 15.3 (15.1)\\
2.0 & 8.1 (6.6)& 10.3 (9.2)& 12.3 (11.3)& 14.1 (13.3)& 15.9 (15.1)\\
3.0 & 9.9 (6.6)& 11.8 (9.2)& 13.6 (11.3)& 15.3 (13.3)& 17.0 (15.1)\\
4.0 & 12.1 (6.6)& 13.8 (9.2)& 15.4 (11.3)& 17.0 (13.3)& 18.5 (15.1)\\
5.0 & 14.5 (6.6)& 16.0 (9.2)& 17.5 (11.3)& 19.0 (13.3)& 20.4 (15.1)\\
\hline
\multicolumn{6}{c}{q=0.001}\\
1.0 & 11.1 (10.8)& 14.1 (13.8)& 16.5 (16.3)& 18.7 (18.5)& 20.7 (20.5)\\
2.0 & 12.2 (10.8)& 15.0 (13.8)& 17.3 (16.3)& 19.4 (18.5)& 21.4 (20.5)\\
3.0 & 14.9 (10.8)& 17.0 (13.8)& 19.1 (16.3)& 21.0 (18.5)& 22.9 (20.5)\\
4.0 & 18.4 (10.8)& 20.1 (13.8)& 21.9 (16.3)& 23.5 (18.5)& 25.2 (20.5)\\
5.0 & 22.2 (10.8)& 23.8 (13.8)& 25.3 (16.3)& 26.8 (18.5)& 28.3 (20.5)\\
\hline
\hline

    \end{tabular}
\end{table*}

\subsection{A case study with astronomical data}
\label{sec:caseStudy}
The methods discussed in this section are further illustrated with the data presented in \cite{spence2023} and \cite{bonamente2023} for the spectra of the quasar \1es. In this example, the independent variable $x$ is
wavelength, and $y$ represent the integer number of photons detected at each wavelength; a complete description of the data
is provided in \cite{spence2023} and in \paperOne.

For this application, the reduced model was a two--parameter power--law
distribution over a range $\pm 1$~\AA\ in
wavelength around an expected absorption line from \ovii\ (six--times ionized atomic oxygen) at $\lambda=25.6545~\AA$. Although the reduced model is not constant, the Poisson mean varies
within the range of $\sim$650--800, i.e, by approximately $\leq \pm 10$\% relative to the mean value, and it is
therefore expected that the considerations used for a constant model (see Sects.~\ref{sec:DeltaYStep} and \ref{sec:DeltaYq1}) apply approximately also to these data. 
The full model consists of the addition of a one--parameter nested component that is akin to the one--bin
step function used in the simulations of Sec.~\ref{sec:DeltaYStep}. Specifically, the model is a narrow \texttt{line}
component in the \texttt{SPEX} software \citep{kaastra1996} that consists of a Gaussian distribution with fixed mean and variance, and variable (positive or
negative) normalization, aimed to detect deviations from the baseline model at a given wavelength. Given that the variance is small, this model affects only  one bin, and it
is thus equivalent to the one--bin step model \citep[see discussion of the model in][]{spence2023}.
When this additional model component is used, the fit to the data result in a statistic $\Delta C=6.6$ for $k=1$ additional free parameter in the nested \texttt{line} model component \citep[see Table 6 of][]{spence2023}.

In the absence of systematic error, the parent $\chi^2(1)$ distribution results in a $p$--value of 0.01, or a 1\%
null hypothesis probability that such $\Delta C=6.6$ improvement in the fit is caused by random fluctuations in the data, and not by the need for the additional nested component. This result is customarily reported as a `detection' of the nested component at the 99\% level of probability.
In the presence of systematic errors, the randomized $\chi^2$ distribution $R(\alpha)$ applies instead
of $\chi^2(1)$, at least
approximately, with $\alpha=2 f \sqrt{\mu}$. This distribution is approximated by the $R_L(\alpha,\lambda)$ with
$\lambda=1.5$ the parameter of the Laplace distribution that replaces the Bessel distribution $K(\alpha)$
in the convolution. Assuming a 5\% level of systematic errors ($f=0.05$), consistent
with known uncertainties in the calibration of the instruments used for the collection of those data \citep{spence2023}, and the
average value of the Poisson mean of $\hat{\mu}\simeq 700$, the parameter of the distribution is $\alpha \simeq 2.7$.
The $R_L$ distribution results in a $p$--value of 0.036, or a 3.6\% null hypothesis probability (see also the critical values of Table~\ref{tab:critVal} for comparison). For a 10\% level of systematic errors ($f=0.1$), the 
null hypothesis probability would increase further to 18.1\%.

This case study illustrates the impact that percent--level systematic errors have on the detection of nested
model components in the regression of Poisson data. The approximate method developed in this paper shows
that the \ovii\ line tentatively detected at 99\% significance in \cite{spence2023} may in fact
be the result of fluctuations associated with systematic errors. In fact, its null--hypothesis probability increases to 3.6\% and to 18.1\%, respectively for 5\% and 10\% systematic errors, making it more likely to occur as a random fluctuation.

\section{Discussion and Conclusions}
\label{sec:conclusions}

This paper has presented a new method to include systematic errors
for hypothesis testing of a nested model component in the Poisson regression to 
a parametric model. This is a common task in data analysis, especially for astronomy 
and the physical sciences where data are often in the form of integer counts (e.g., photons) as a
function of one or more independent variables (e.g., time or energy). The main result of this 
paper is that the fit statistic \DeltaC, which was proposed by \cite{cash1979} as the statistic
of choice for nested model components, can be generalized to a \DeltaCSys\ statistic whose null--hypothesis distribution can be either simulated with a simple numerical procedure, or approximated analytically. Accordingly, critical
values of the distribution can be easily calculated for the purpose of hypothesis testing.

The methods presented in this paper are derived from the general framework for the inclusion
of systematic errors presented in \paperOne, where the distribution of the 
\gof\ statistic \cmin\ in the presence
of systematic errors was presented. For the \DeltaC\ statistic, we derived exact results in the
case of a simple constant parent model for $k=1$ additional parameter, and surmised that the results
can be generalized to more complex models and to the case of $k>1$ parameters. 
The approximate parent distribution of \DeltaCSys\ is referred to as the randomized chi--squared distribution
$R(\nu,\alpha)$, and it is due to the convolution of a $\chi^2(\nu)$ distribution that applies in the absence of systematic errors (i.e, the distribution of the $\Delta X$ statistic) , and a Bessel distribution $K_0(\alpha)$ that models the novel contribution
to the statistic provided by systematics (the $\Delta Y$ statistic).
Extensive numerical simulations were also performed that 
support the proposed results, and an analytic approximation $R_L(\nu,\alpha,\lambda)$ that is made
possible by the approximation of the Bessel distribution with a Laplace (or double--exponential) distribution with additional parameter $\lambda$. This analytic approximation is especially convenient in terms of computational speed for quantiles and other properties of the distribution.

Although the model of systematic errors can be applied to any type of regression, the results presented in this paper apply to the extensive ($N \gg 1$ data points) and 
large--mean case ($\mu \gg 1$), same as in \paperOne. In particular, in Sec.~\ref{sec:DeltaYGeneral}
we have assumed that the distribution of choice for the $M_i$ variables that determine
the distribution of systematic errors is Gaussian. As also discussed in \paperOne, this
assumption is only meaningful when $f\ll 1$ (which is a general feature of the model) and when the
parent mean is large, so that negative values in the distribution are unlikely. In fact, negative
values for $M_i$ are not tenable, since $M_i$ represents the mean of a Poisson distribution that 
can obviously only be non--negative. For applications in the low--mean regime, which are not
considered in this paper, different distributions (such as the positive--definite  gamma) must therefore be used, which will lead to modifications to the results presented in this paper.

The effect of systematic error on the \DeltaC\ statisic can be significant even for
a moderate level of systematics, e.g., at the few percent level. The results presented in this paper show that critical
values of the \DeltaCSys\ statistic, using the proposed randomized chi--squared distribution, are significantly larger than those using the traditional chi--squared distribution, as we illustrated
with the case of a previously claimed detection of an absorption line by \cite{spence2023}. 
For the type of
Poisson applications discussed in this paper and that are especially common in astrophysics,
it is therefore recommended that the assessment of the presence of a nested model component 
be carried out with the proper \DeltaC\ statistic,
and that the effect of known levels of systematics be included by considering the distributions of
\DeltaCSys\
presented in this paper.

\appendix
\section{Approximations}

\subsection{Approximations of the $\Delta Y$ statistic}
\label{app:randomization}
Starting with \eqref{eq:YStat}, the $\Delta Y$ statistic can be approximated as
\[
\Delta Y = 2 \sum\limits_{i=1}^N (x_{r,i} -x_i) -y_i\left( \dfrac{x_{r,i}}{\hat{\mu}_{r,i}} - \dfrac{x_i}{\hat{\mu}_i}\right)
+ \dfrac{y_i}{2}\left( \left(\dfrac{x_{r,i}}{\hat{\mu}_{r,i}}\right)^2 - \left(\dfrac{x_i}{\hat{\mu}_i}\right)^2 \right) + \dots
\]
where $\hat{\mu}_i$ and $\hat{\mu}_{r,i}$ have been  randomized simultaneously
according to \eqref{eq:mui}.
The effect of this randomization of the two variables $M_i$ and $M_{r_i}$
is to make
\[
\dfrac{x_{r,i}}{\hat{\mu}_{r,i}} = \dfrac{x_i}{\hat{\mu}_i} = \beta\, f,
\]
where $\beta \in \mathbb{R}$ is a number that represents the randomization of the $M_i$ and $M_{r,i}$
random variables, and it is the \emph{same} for both variables;
and $f \ll 1$ is the constant fractional systematic error according to \eqref{eq:f}. Therefore the $\Delta Y$ statistic becomes
\[
\Delta Y = 2 \sum\limits_{i=1}^N (x_{r,i} -x_i).
\]
This equation holds regardless of the number of terms used in the logarithmic expansion, and it is a result of the
type of randomization made for the $M_i$ and $M_{r_i}$ random variables that models the presence of systematic errors.

\subsection{Approximation of the Bessel distribution with a Laplace distribution and convolution with a chi--squared distribution}
\label{app:Laplace}

With the Laplace distribution \eqref{eq:fLaplace} replacing the Bessel distribution \eqref{eq:fDeltaC}, it is now possible to proceed to the convolution \eqref{eq:fDeltaC} between a chi--squared $\chi^2(\nu)$ and 
a Laplace $\Laplace(\alpha/\lambda)$ distribution. The two distributions are
\begin{equation}
    \begin{cases}
        f_{\chi^2}(x; \nu)=\left(\dfrac{1}{2}\right)^{\nu/2} \dfrac{1}{\Gamma(\nu/2)} e^{-x/2}x^{\nu/2-1}\; \text{ for } x \geq 0,\\[10pt]
        f_{\text{L}}(x; \alpha, \lambda) = \dfrac{\lambda}{2 \alpha}e^{-\dfrac{\lambda |x|}{\alpha}} \; \text{ for } x \in \mathbb{R},
    \end{cases}
\end{equation}
and the convolution leads to the probability distribution of
\DeltaCSys, 
\begin{equation}
\begin{aligned}
     f_{R_L}(z)= 
     \dfrac{\lambda/2}{\alpha\, \Gamma(\nu/2)\,  2^{\nu/2}} \times
    \begin{cases}
        \displaystyle\int_{0}^{\infty} e^{-\lambda \left(\dfrac{y-z}{\alpha}\right)} y^{\nu/2-1} e^{-y/2} dy 
        & \text{for } \; z<0 \\[10pt] 
       \displaystyle\int_{0}^{z} e^{-\lambda \left(\dfrac{z-y}{\alpha}\right)} y^{\nu/2-1} e^{-y/2} dy +
       \displaystyle\int_{z}^{\infty} e^{-\lambda \left(\dfrac{y-z}{\alpha}\right)} y^{\nu/2-1} e^{-y/2} dy & \text{for } z\geq0, 
    \end{cases}
    \end{aligned}
\end{equation}
 where the absolute value required the separation of the integral for $z>0$, respectively for the range
 of integration $y<z$
 and $y\geq z$. The approximation of the Bessel distribution with a Laplace distribution
therefore makes it possible to write the integral above as
\begin{equation}
\begin{aligned}
     f_{R_L}(z)= 
     \dfrac{\lambda/2}{\alpha\, \Gamma(\nu/2)\,  2^{\nu/2}} \times
     \begin{cases}
       e^{\displaystyle cz} \underbracket{\displaystyle\int_{0}^{\infty} e^{\displaystyle -c_1 y}\, y^{\nu/2-1}dy}_{I_1} 
        & \text{for } \; z<0 \\[10pt] 
       e^{\displaystyle -cz} \underbracket{\displaystyle\int_{0}^{z} e^{\displaystyle -c_2 y}\, y^{\nu/2-1}dy}_{I_2} +
       e^{\displaystyle cz} \underbracket{\displaystyle\int_{z}^{\infty} e^{\displaystyle -c_1 y}\, y^{\nu/2-1} dy}_{I_3} & \text{for } z\geq0, 
    \end{cases}
    \end{aligned}
    \label{eq:fDeltaCApp}
\end{equation}
where $c=\lambda/\alpha>0$ and $c_1=c+1/2>0$ are two positive constants, and $c_2=1/2-c$ can have either sign.
The three integrals in  \eqref{eq:fDeltaC} can be evaluated as follows. The two integrals $I_1$ and $I_3$ are immediately found for all values of $\nu$, 

\[ I_1= \displaystyle\int_{0}^{\infty} e^{\displaystyle -c_1 y}\, y^{\nu/2-1}dy
= \dfrac{\Gamma(\nu/2)}{c_1^{\nu/2}} 
\]
 and likewise
\[
I_3=\displaystyle\int_{z}^{\infty} e^{\displaystyle -c_1 y}\, y^{\nu/2-1} dy=\dfrac{\Gamma(\nu/2,c_1 z)}{c_1^{\nu/2}}
\]
where $\Gamma(s,x)$ is the upper incomplete
 gamma function. For $I_2$,
\def\erfi{\text{erfi}}
\def\erf{\text{erf}}
 \[
I_2= \displaystyle\int_{0}^{z} e^{\displaystyle -c_2 y}\, y^{\nu/2-1}dy=
\begin{cases}
\dfrac{\gamma(\nu/2,c_2 z)}{c_2^{\nu/2}}, \, \text{ if } c_2>0, \text{ or } \lambda/\alpha<1/2\\[10pt]    
\dfrac{z^{\nu/2}}{\nu/2},   \, \text{ if } c_2 = 0 \text{ or } \lambda/\alpha = 1/2\\[10pt]
I_{2,n}(\nu), 
\, \text{ if } c_2 < 0, \text{ or } \lambda/\alpha > 1/2\\
\end{cases}
 \]
 where $\gamma(s,x)$ is the lower incomplete gamma function, and $I_{2,n}(\nu)$ must be evaluated as
 a function of $\nu$ for negative values of $c_2$ that make the exponent positive in the integrand.
For $\nu=1$, the integral $I_{2,n}(\nu)$ is 
\[
I_{2,n}(\nu=1) = \int_{0}^z e^{y |c_2|} y^{-1/2} dy = \sqrt{\dfrac{\pi}{|c_2|}} \cdot  \erfi(\sqrt{z |c_2|})
\]
where $\erfi(x)=-i\, \erf(ix)$ is the imaginary error function, with $\erf(x)$ the usual
 error function (see Appendix~\ref{app:math} for relevant integrals).
Analytic expressions for $I_{2,n}(\nu)$ for $\nu>1$ can be obtained upon integration by parts and recursion. The first few integrals are reported below:
\[
\begin{cases}
  I_{2,n}(\nu=2) =   \int_{0}^z e^{y |c_2|}\, dy = \left. \dfrac{e^{y |c_2|}}{|c_2|}\right|_0^z;\\[10pt]
  I_{2,n}(\nu=3) = \int_{0}^z e^{y |c_2|}\, y^{1/2} dy = \left. \dfrac{e^{y |c_2|} y^{1/2}}{|c_2|}\right|_0^z - \dfrac{1}{2 |c_2|} I_{2,n}(\nu=1);\\[10pt]
  I_{2,n}(\nu=4) = \int_{0}^z e^{y |c_2|}\, y dy = \left. \dfrac{e^{y |c_2|} y }{|c_2|}\right|_0^z 
  - \dfrac{1}{|c_2|} I_{2,n}(\nu=2);\\[10pt]
  I_{2,n}(\nu=5) = \int_{0}^z e^{y |c_2|}\, y^{3/2} dy = \left. \dfrac{e^{y |c_2|} y^{3/2} }{|c_2|}\right|_0^z - \dfrac{3}{2 |c_2|} I_{2,n}(\nu=3);\\[10pt]
  \text{etc.}
\end{cases}
\]

 The approximate distribution for $\Delta C$ is therefore
 \begin{equation}
    f_{R_L}(z; \nu, \alpha, \lambda)= \dfrac{\lambda/2}{\alpha \sqrt{2 \pi}} \times
    \begin{cases} 
    \begin{aligned} & e^{\displaystyle cz} I_1 
        & \text{for } \; z<0 \\[5pt] 
       & e^{\displaystyle -cz}\, I_2 + e^{\displaystyle cz}\, I_3
      & \text{for } z\geq0.
       \end{aligned}
    \end{cases}
    \label{eq:randChiLaplace}
\end{equation}
where $\nu \in \mathbb{N}$ is the number of \dof\ of the $\Delta X \sim \chi^2(\nu)$ variable, $\alpha \geq 0$ a real number representing the parameter of the Bessel distribution that models the overdispersion $\Delta Y \sim K_0(\alpha)$, and $\lambda$ a fixed parameter that is used to approximate
the Bessel distribution with a Laplace distribution, typically $\lambda =1.5$
Fig.~\ref{fig:randChiExp} illustrates the $R(\nu,\alpha)$ and $R_L(\nu,\alpha,\lambda=15)$ 
distribution for representative values of the parameters.

 \begin{figure}
     \centering
     \includegraphics[width=3.5in]{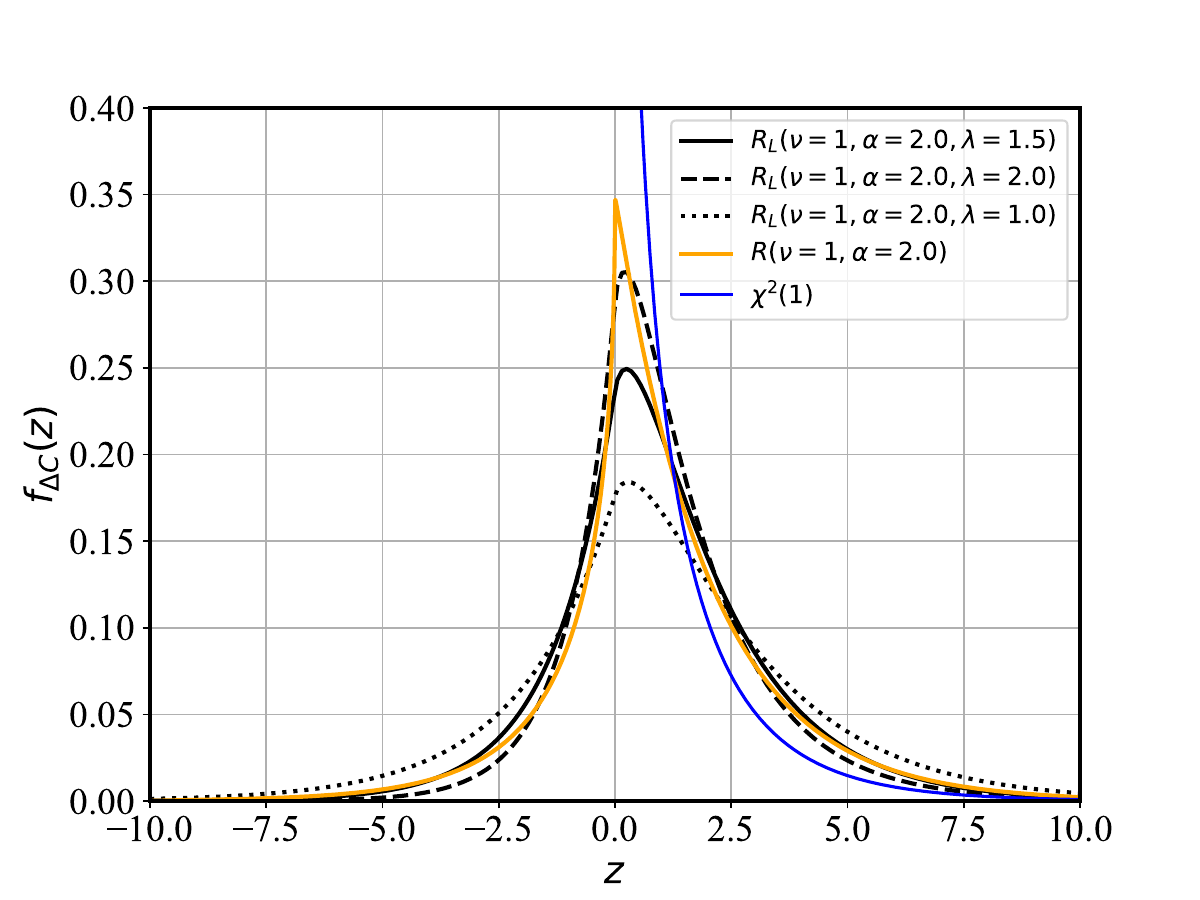}
     \includegraphics[width=3.5in]{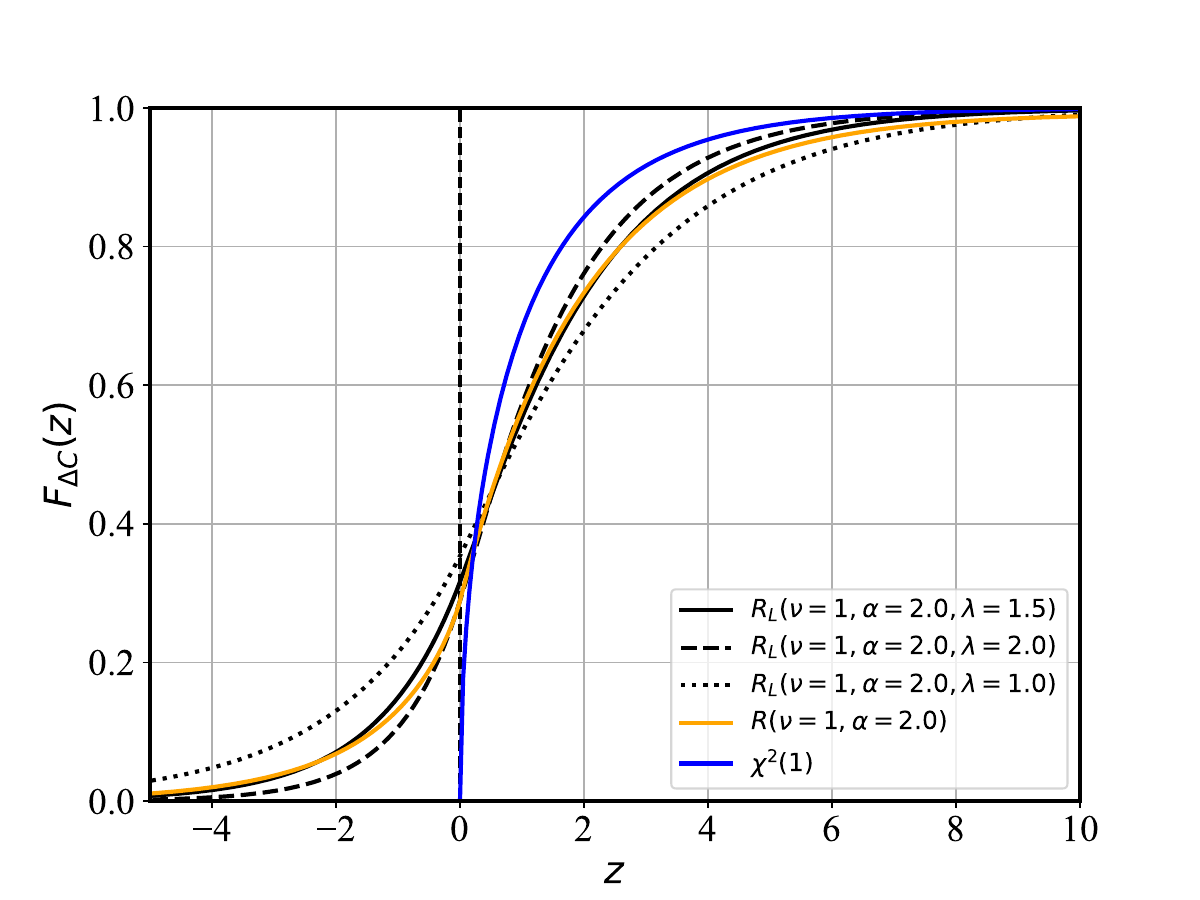}
     \caption{Illustration of the probability density function (left) and the cumulative
     distribution (right) of the randomized $\chi^2$ distribution $R(\nu, \alpha)$, along with the approximation
     $R_L(\nu,\alpha)$ that uses
     the Laplace approximation to the Bessel distribution to yield an analytic density. This distribution is intended as the parent
     distribution for the $\Delta C$ statistic for model component with one additional
     nested parameter ($\nu=1$).}
     \label{fig:randChiExp}
 \end{figure}

\section{A mathematical conjecture on the general distribution of $\Delta Y$}
\label{app:conjecture}
Under the null hypothesis that the reduced model is the parent model, a full model with a nested
component featuring $k\geq 1$ adjustable parameters results in $\Delta X \sim \chi^2(k)$, according to the
Wilks theorem described in Sec.~\ref{sec:DeltaC}. This result applies asymptotically to any model
parameterization. Another general result is the approximation \eqref{eq:DeltaY}
for the $\Delta Y$ statistic, which
applies to small values of the intrinsic model variance, $f \ll 1$, regardless of
model parameterization.
The starting point towards a generalization of Theorem~\ref{th:DeltaY} is the approximation \eqref{eq:DeltaXLemma1} for $\Delta X$, 
\[
            \Delta X \simeq \sum_{i=1}^N \dfrac{ \Delta \hat{\mu}_i^2}{\hat{\mu}_{r,i}} \sim \chi^2(k)
\]
with $\Delta \hat{\mu}_i = \hat{\mu}_{r,i}-\hat{\mu}_i$. In general, the parent
reduced model means $\mu_{r,i}$ are not constant, thus preventing the straightforward parameterization that leads to 
\eqref{eq:DeltaXLemma1} and therefore the proof of Lemma~\ref{lemma1}. Instead, we propose the following:

\begin{conjecture}[Reparameterization of $\Delta X$ in $k$ independent terms]
\label{conjectureDeltaX}
For a dataset with $N$ independent points, and a full model with $k$ additional parameters
in a nested component, it is possible to
find a re--parameterization of $\Delta X$ in \eqref{eq:DeltaXApprox} such that
\begin{equation}
    \Delta X = \sum_{j=1}^k \dfrac{(\Delta \hat{\overline{\mu}}_j)^2}{\overline{\mu}_{r,j}} 
    \coloneqq \sum_{j=1}^k \dfrac{\overline{X}_j^2}{\overline{\mu}_{r,j}} 
    \simeq \dfrac{1}{\overline{\mu}_r} \sum_{j=1}^k \overline{X}_j^2
    \label{eq:DeltaCLemma2}
\end{equation}
where $\Delta {\overline{\mu}}_j= {\overline{\mu}}_{r,j} - {\overline{\mu}}_j \coloneqq \overline{X}_j$ is 
the 
difference between suitable averages of the parent means of the reduced and full models.
\end{conjecture}

Conjecture~\ref{conjectureDeltaX} posits that the differences between the full and the reduced
models in the $N$ bins can be expressed as if they were  
concentrated in $k$ \emph{independent} bins, 
as was the case for the one--bin model of Sec.~\ref{sec:DeltaYq1}. The conjecture is based on the known property of a $\chi^2(k)$ random variable, which can be written as the sum of the squares of $k$ 
independent and identically distributed standard normal distributions. Therefore, each of the $k$ terms in \eqref{eq:DeltaCLemma2} 
must be such that
\[
\dfrac{\overline{X}_j^2}{\overline{\mu}_{r,j}}  \sim \chi^2(1) 
\]
thus implying that
\begin{equation}
    \overline{X}_j =\Delta {\overline{\mu}}_j \sim N(0, \overline{\mu}_{r,j})\; \text{ for } j=1,\dots,k.
\label{eq:Xj}
\end{equation}

The conjecture thus consists in the identification of suitable averages of the means, which are indicated as ${\overline{\mu}}_{r,j}$ and 
${\overline{\mu}}_j$. Moreover, provided the means $\overline{\mu}_{r,j}$ are sufficiently similar to one another, it ought to be possible to identify an overall average (reduced) mean $\overline{\mu}_r$ that therefore makes the statistical problem 
identical to that of \eqref{eq:DeltaXLemma1}.

If Conjecture~\ref{conjectureDeltaX} applies, it would therefore follow that:
\begin{corollary}[General distribution of sum of model deviations]
\label{th:corollary1}
Under the null hypothesis that the parent model is the reduced model, and that the full model
has $k \geq 1$ additional free parameters in a nested component, then the sum of the deviations is
\begin{equation}
\sum_{i=1}^N (\hat{\mu}_{r,i}-\hat{\mu}_i) \sim N(0, k \cdot \overline{\mu}_r).
\end{equation}
in the asymptotic limit of large parent Poisson means, where $\overline{\mu}_r$ is a suitable average of
the parent Poisson mean for the data, as surmised in Conjecture~\ref{conjectureDeltaX}.
\end{corollary}
Corollary~\ref{th:corollary1} would be proven using the same methods as Lemma~\ref{lemma1}.
If Corollary~\ref{th:corollary1} holds for any model parameterization, then it would immediately be possible to prove the
following:
\begin{corollary}[General distribution of $\Delta Y$] Under the same assumptions as corollary~\ref{th:corollary1},
\label{th:corollary2}
    \begin{equation}
        \Delta Y = 2 \sum_{i=1}^N (x_{r,i}-x_i) \sim K_0( \alpha_k),
    \end{equation}
    with parameter
    \begin{equation}
        \alpha_k= 2 f \sqrt{k\, \overline{\mu}_r}
        \label{eq:alphaq}
    \end{equation}
    where $k$ is the number of free parameters in the nested model
    component.
\end{corollary}
Corollary~\ref{th:corollary2} would be proven as a direct consequence of corollary~\ref{th:corollary1}, following the same proof as Theorem~\ref{th:DeltaY}, and it would generalize Theorem~\ref{th:DeltaY} for $k>1$ and for any model parameterization. 
The results of Sec.~\ref{sec:DeltaYq1} are only established for a simple
constant baseline model, and for a nested model component with one additional parameter. The proposed generalization to any parameterization of the reduced model and $k>1$ rely on the applicability of Conjecture~\ref{conjectureDeltaX}.

\section{List of integrals}
\label{app:math}
The gamma function is defined as
\begin{equation}
    \Gamma(n+1) = \int_{0}^{\infty} x^n e-x \,dx \text{ for } n>-1
\end{equation}
with
\begin{equation}
    \int_{0}^{\infty} x^n e^{-a x} dx = \dfrac{\Gamma(n+1)}{a^{n+1}}.
\end{equation}
The lower and upper incomplete gamma functions are defined respectively as
\begin{equation}
    \gamma(s,x)=  \int_{0}^{x} x^{s-1} e^{-x} dx
\end{equation}
and
\begin{equation}
    \Gamma(s,x) =\int_{x}^{\infty} x^{s-1} e^{-x} dx
\end{equation}
with
\[ \Gamma(s) =  \gamma(s,x) + \Gamma(s,x).
\]

The error function is defined as
\begin{equation}
    \erf(x) = \dfrac{2}{\pi} \int_0^x e^{\displaystyle -t^2} dt
\end{equation}
with the imaginary error function defined as 
\begin{equation}
    \erfi(x) = -i\, \erf(i\,x) = \dfrac{2}{\pi} \int_0^x e^{\displaystyle t^2} dt.
\end{equation}


\bibliographystyle{aasjournal}


\end{document}